\documentclass[11pt]{amsart}
\usepackage[percent]{overpic}
\usepackage[applemac]{inputenc}
\usepackage[titletoc]{appendix}
\usepackage{todonotes}
\usepackage{blindtext}
\usepackage{amssymb}
\usepackage{amsfonts}
\usepackage{amsthm}
\usepackage{tikz}
\usepackage{tikz-cd}
\usepackage{caption}
\usepackage{graphicx}
\usepackage{xcolor}
\usepackage{shadethm}
\usepackage{subfigure}
\usepackage[normalem]{ulem}
\usepackage{epigraph}

\usepackage{enumerate}
\usepackage{verbatim}
\usetikzlibrary{trees}
\usetikzlibrary{decorations.markings}
\usetikzlibrary{arrows}
\usepackage{mathrsfs}
\usepackage[margin=1.5in]{geometry}
\usepackage{xparse}
\usetikzlibrary{positioning,arrows,patterns}
\usetikzlibrary{decorations.markings}
\usetikzlibrary{calc}
\tikzset{
  photon/.style={decorate, decoration={snake}, draw=black},
  fermion/.style={draw=black, postaction={decorate},decoration={markings,mark=at position .55 with {\arrow{>}}}},
  vertex/.style={draw,shape=circle,fill=black,minimum size=5pt,inner sep=0pt},
particle/.style={thick,draw=black},
particle2/.style={thick,draw=blue},
avector/.style={thick,draw=black, postaction={decorate},
    decoration={markings,mark=at position 1 with {\arrow[black]{triangle 45}}}},
gluon/.style={decorate, draw=black,
    decoration={coil,aspect=0}}
 }
\NewDocumentCommand\semiloop{O{black}mmmO{}O{above}}
{%
\draw[#1] let \p1 = ($(#3)-(#2)$) in (#3) arc (#4:({#4+180}):({0.5*veclen(\x1,\y1)})node[midway, #6] {#5};)
}
\tikzset{
  negate/.style={
    decoration={
      markings,
      mark= at position 0.5 with {
        \node[transform shape] (tempnode) {$/$};
      },
    },
    postaction={decorate},
  },
}

\newcommand{\ct}{\mathrm{count}}

\usepackage{url}
\usepackage{hyperref}
\hypersetup{colorlinks=true, linktocpage, linkcolor=blue, citecolor=red, filecolor=green, urlcolor=cyan, pdfborderstyle={/S/U/W 1}}

\theoremstyle{plain}
\newtheorem{thm}{Theorem}[section]
\theoremstyle{definition}
\newtheorem{lem}[thm]{Lemma}
\newtheorem{rem}[thm]{Remark}
\newtheorem{ex}[thm]{Example}
\newtheorem{prop}[thm]{Proposition}
\newtheorem{cor}[thm]{Corollary}

\newtheorem{defn}[thm]{Definition}

\theoremstyle{remark}

\newcommand{\R}{\mathbb{R}}

\newcommand{\C}{\mathbb{C}}

\newcommand{\id}{\mathrm{id}}
\newcommand{\calH}{\mathcal{H}}
\newcommand{\calS}{\mathcal{S}}
\newcommand{\calC}{\mathcal{C}}
\newcommand{\calK}{\mathcal{K}}
\newcommand{\calG}{\mathcal{G}}

\def\gpd{\,\lower1pt\hbox{$\longrightarrow$}\hskip-.24in\raise2pt
               \hbox{$\longrightarrow$}\,}

\newcommand{\Z}{\mathbb{Z}}

\let\ul=\underline

\numberwithin{equation}{section}
\makeatletter

\pgfdeclareshape{genus}{
 \anchor{center}{\pgfpointorigin}
\backgroundpath{
    \begingroup
    \tikz@mode
    \iftikz@mode@fill

         \pgfpathmoveto{\pgfqpoint{-0.78cm}{-.17cm}}
     \pgfpathcurveto %
            {\pgfpoint{-0.35cm}{-.44cm}}
        {\pgfpoint{0.35cm}{-.44cm}}
        {\pgfpoint{.78cm}{-0.17cm}} 
     \pgfpathmoveto{\pgfqpoint{-0.78cm}{-0.17cm}}
     \pgfpathcurveto %
            {\pgfpoint{-0.25cm}{.25cm}}
        {\pgfpoint{.25cm}{.25cm}}
        {\pgfpoint{0.78cm}{-0.17cm}}
        \pgfusepath{fill}
        \fi

        \iftikz@mode@draw
     \pgfpathmoveto{\pgfqpoint{-1cm}{0cm}}
     \pgfpathcurveto %
            {\pgfpoint{-0.5cm}{-.5cm}}
        {\pgfpoint{0.5cm}{-.5cm}}
        {\pgfpoint{1cm}{0cm}}

     \pgfpathmoveto{\pgfqpoint{-0.75cm}{-0.15cm}}
     \pgfpathcurveto %
            {\pgfpoint{-0.25cm}{.25cm}}
        {\pgfpoint{.25cm}{.25cm}}
        {\pgfpoint{0.75cm}{-0.15cm}}
              \pgfusepath{stroke}
        \fi
        \endgroup
    }
    }

    \makeatother

\begin{document}
\makeatletter
\providecommand\@dotsep{5}
\def\listtodoname{List of Todos}
\def\listoftodos{\@starttoc{tdo}\listtodoname}
\makeatother

\title[Convolution algebras for Relational Groupoids and Reduction]{Convolution algebras for Relational Groupoids and Reduction}
\author[I. Contreras]{Ivan Contreras}
\author[N. Moshayedi]{Nima Moshayedi}
\author[K. Wernli]{Konstantin Wernli}
\address{Department of Mathematics and Statistics\\Amherst College\\31 Quadrangle Drive, Amherst MA 01002}
\email[I.~Contreras]{icontreraspalacios@amherst.edu}
\address{Institut f\"ur Mathematik\\ Universit\"at Z\"urich\\ 
Winterthurerstrasse 190
CH-8057 Z\"urich}
\email[N.~Moshayedi]{nima.moshayedi@math.uzh.ch}
\address{Department of Mathematics\\University of Notre Dame\\ 225 Hurley, Notre Dame IN}
\email[K.~Wernli]{kwernli@nd.edu}

\subjclass[2010]{
18B10, %Category of relations, additive relations
18B40, %Groupoids, semigroupoids, semigroups, groups (viewed as categories)
18D35, %Structured objects in a category (group objects, etc.)
%18G30, %Simplicial sets, simplicial objects (in a category)
20L05, %Groupoids (i.e. small categories in which all morphisms are isomorphisms)
57R56%Topological quantum field theories
} 
\maketitle

\begin{abstract} We introduce the notions of relational groupoids and relational convolution algebras. We provide various examples arising from the group algebra of a group $G$ and a given normal subgroup $H$. We also give conditions for the existence of a Haar system of measures on a relational groupoid compatible with the convolution, and we prove a reduction theorem that recovers the usual convolution of a Lie groupoid.
% Moreover, we provide these notes with some examples.

\end{abstract}
\tableofcontents

\section{Introduction}
\subsection{Motivation}
Symplectic groupoids are fundamental objects in Poisson geometry. Every symplectic groupoid $G \rightrightarrows M $ induces a Poisson structure on $M$ \cite{Coste}. Such Poisson manifolds are called \emph{integrable}, and $G$ is also called a \emph{symplectic realization} of $M$. It is a well-known fact that not all Poisson manifolds are integrable and that there are explicit obstructions to the integration \cite{CrainicFernandes}. However, one can associate to every Poisson manifold $M$ a \emph{relational symplectic groupoid} \cite{Relational, ContrerasThesis} which is
%Relational symplectic groupoids were introduced in \cite{Relational, ContrerasThesis} in order to describe the groupoid structure of
%the phase space of a 2-dimensional topological field theory %(TFT),
%the Poisson Sigma Model (PSM) %\cite{Ikeda,SchallerStrobl,CF3}, \emph{before} gauge reduction \cite{CF2}. 
%In this context, the relational symplecic groupoid associated to a Poisson manifold can be thought of as 
an infinite-dimensional symplectic manifold, %integration of that Poisson manifold. It is 
equipped with Lagrangian submanifolds
%(evolution relations) 
that model the structure maps of a symplectic groupoid. 
%% In this version there is no field theory until a bit later, so people don't get scared. 
%On the other hand, i
In \cite{Hawkins} Hawkins showed that a Poisson manifold can be quantized\footnote{Loosely speaking, a ``quantization'' of a Poisson manifold is a non-commutative deformation of the algebra of functions, or a Lie subalgebra of it, subject to a subset of certain axioms put forward by Dirac \cite{Dirac}.}
via a twisted polarized convolution $C^*$-algebra of a symplectic groupoid integrating that Poisson manifold. 
%\textcolor{cyan}{Many examples show that this approach unifies the construction of geometric quantization of symplectic manifolds with the $C^*$-algebra quantization of a Lie groupoid.}\marginpar{Maybe we can skip this sentence, or be a bit more precise}

The first main idea behind this paper is to generalize Hawkins' approach to arbitrary Poisson manifolds (integrable or not) by generalizing this construction to relational symplectic groupoids. 
In order to achieve this objective, we introduce the notion of \emph{relational groupoids} and their corresponding \emph{relational convolution algebras}, which is an analogue of the convolution algebra in the partial category $\mathbf{Rel}$ of sets and relations. 
%\marginpar{I realized I don't understand this paragraph. Also, I am not sure it adds to the main idea?}

We study various examples of relational convolution algebras that arise from extending Haar systems of measures to relational groupoids, and we prove our main result: a \emph{reduction theorem for relational convolution algebras}, which recovers the usual groupoid convolution algebra. This is also the first step towards proving the ``quantization commutes with reduction" conjecture by Guillemin and Sternberg \cite{GS} in the 
%case of the PSM.
setting of groupoid quantization. 

In particular, this result serves as the first step towards reduction of its quantization (convolution algebras for relational groupoids). The next step is to construct the polarized algebra for relational symplectic groupoids. In addition to this, we hope to use relational convolution algebras to recover the C*-algebra quantization of Poisson pencils via reduction, recovering the results obtained in \cite{Bonechi} regarding the Bohr-Sommerfeld groupoid.

The second main idea behind this paper is that the relational symplectic groupoids could be used to study the relation between groupoid quantization and deformation quantization in a field-theoretic way, as follows.
Relational symplectic groupoids were introduced in \cite{Relational, ContrerasThesis} in order to describe the groupoid structure of
the phase space of a 2-dimensional topological field theory,
the \emph{Poisson Sigma Model} (PSM) \cite{Ikeda,SchallerStrobl,CF3}, \emph{before} gauge reduction \cite{CF2}.
In \cite{CF1}, Cattaneo and Felder have shown that the perturbative quantization of the PSM using the Batalin--Vilkovisky (BV) formalism \cite{BV1,BV2,BV3} yields Kontsevich's star product \cite{Kontsevich}, a deformation quantization associated to any Poisson manifold. 

%\textcolor{magenta}{
It was shown by Cattaneo, Mnev and Reshetikhin in \cite{CMR2}, that the BV formalism can be extended to deal with the perturbative quantization of gauge theories on manifolds with boundary by coupling the Lagrangian approach of the Batalin--Vilkovisky construction \cite{BV1,BV2,BV3} in the bulk to the Hamiltonian approach of the Batalin--Fradkin--Vilkovisky construction \cite{BF1,FV1,Sta} on the boundary. This is known today as the \emph{BV-BFV formalism} \cite{CattMosh}.
Recently \cite{RSG_Quantization}, this formalism has been applied to the relational symplectic groupoid for constant Poisson structures, linking the BV-BFV perturbative quantization of the relational symplectic groupoid and Kontsevich's star product in this case by methods of cutting and gluing for Lagrangian evolution relations. 
%}
% providing a BV-BFV field-theoretic interpretation of Konstevich's star product \cite{Kontsevich} by using cutting and gluing techniques,
%for constant Poisson structures K: (this is double)
% recovering a global \emph{Moyal quantization} \cite{Moyal}.
% and \emph{Fedosov quantization} \cite{Fedosov}). K: I'm not sure we do that, 
% N: Not exactly, but it is related. The global Poisson case is motivated by Fedosov's approach (symplectic case --> const. Poisson structure --> Moyal quant). 

%In particular, one considers the PSM and the fact that its perturbative quantization yields Kontsevich's star product in the asymptotic limit \cite{CF1}. 

These constructions have been partially extended to a wider class of Poisson structures and source manifolds in \cite{CMW2} and 
%general source manifolds  
more general AKSZ theories \cite{AKSZ} in \cite{CMW}. We expect these results to be generalized to yield a BV-BFV description of a \emph{global} deformation quantization for general Poisson manifolds, not necessarily Kontsevich's star product, which might produce some interesting algebraic structures. %This provides a field-theoretic approach to quantizing a Poisson manifold.

%\textcolor{cyan}{  
A clear and explicit connection between geometric quantization (in terms of $C^{*}$-algebras) of the reduced phase space and deformation quantization of Poisson manifolds, via the PSM \cite{CF1}, remains an open question.
By constructing relational convolution algebras for the PSM, we hope to connect Kontsevich's and Hawkins' approaches via BV-BFV quantization of the relational symplectic groupoid in the future.
Eventually, these techniques might also help to generalize Kontsevich's star product to higher genera. In particular, the convolution algebra of a relational symplectic groupoid is the first step towards prescribing a field-theoretic interpretation of the $C^*$-algebra quantization of Poisson manifolds, in terms of the non perturbative PSM \cite{PSM_GQ}. 
%}
%\marginpar{Maybe we can merge this partly with the ``second main idea'' above to avoid repetition.}
%\marginpar{I agree, can you give it a try and we see how it looks?}
%\marginpar{\textcolor{magenta}{I would leave the cyan paragraph as it is.}}

Another motivation to this paper is the connection between groupoids and Frobenius objects in a dagger monoidal category. For instance, a representative example of a relational convolution algebra is the relational group algebra, a version \emph{up to equivalence}, of the group algebra of a group $G$. Group algebras are particular cases of Frobenius algebras, so relational convolution algebras provide a new class of examples of Frobenius objects in the category of sets and relations, which are also in correspondence with groupoids \cite{Frobenius1, Frobenius2}. In a work in preparation \cite{MehtaKeller}, we study Frobenius objects arising from groupoids in the category of spans, via simplicial sets. 

\subsection{Notation and conventions}
We will denote groups or groupoids by usual letters $G,H,K$ and relational groups or relational groupoids by calligraphic letters $\calG,\calH,\calK$.
Moreover, we will use Greek letters to denote elements in the set of morphisms of a groupoid. Latin letters $g,h,k$ (or $g_1,g_2,\ldots$) will be used for elements of a relational groupoid $(\mathcal{G},L,I)$. We will put an underline for a (relational) group(oid) to express the corresponding space observed after reduction. 
Underlined Latin letters $\underline{g},\underline{h},\underline{k}$ (or $\underline{g}_1,\underline{g}_2,\ldots$) will denote that the given object is obtained by reduction. 
A slashed arrow between two sets
$A \nrightarrow B$ denotes a relation from $A$ to $B$, i.e. a subset of $A\times B$. In this manuscript we treat relations as subsets of the Cartesian product, and the domain and codomain of the relation are prescribed in case it is ambiguous. Latin letters $x,y,z$ (or $x_1,x_2,\ldots$) will denote elements in the space of objects (base) of a (relational) groupoid. Functions will be denoted by $f_1,f_2,\ldots$ to avoid confusion with elements of a (relational) group(oid). 

\subsection*{Acknowledgements} We would like to thank Iakovos Androulidakis, Alberto Cattaneo, Eli Hawkins, Rajan Mehta and Michele Schiavina for useful discussions and comments on the manuscript. We also thank the anonymous referee for their useful and detailed comments.  N. M. was supported by the NCCR SwissMAP, funded by the Swiss National Science
Foundation, and by the SNF grant No. 200020\_192080. N. M. would like to thank the University of Illinois Urbana-Champaign for hospitality where this work started.
K. W. would like to thank the University of Z\"urich where a part of this work was
written, and acknowledges partial support of NCCR SwissMAP, funded by the Swiss National Science Foundation, and by the COST Action MP1405 QSPACE, supported by COST (European Cooperation in Science and Technology), and the SNF grant No. 200020 172498/1 during his affilitation with the University of Zurich. K. W. acknowledges further support from a BMS Dirichlet postdoctoral fellowship and the SNF Postdoc.Mobility grant P2ZHP2\_184083, and would like to thank the Humboldt-Universit\"at Berlin, in particular the group of Dirk Kreimer, and the University of Notre Dame for their hospitality.
% \marginpar{\textcolor{magenta}{Konstantin, you were not supported by this SNF grant. This is the new grant which has started in April 2020.}}

\section{Background material}
\label{sec:Background_material}
\subsection{Groupoids}
\label{subsec:Groupoids}
Recall that a \emph{groupoid} is a small category whose morphisms are invertible. We denote a groupoid by $G\rightrightarrows M$, endowed with source map $s\colon G\to M$ and target map $t\colon G\to M$, where $G$ is the set of morphisms and $M$ the set of objects. We denote by $G^{(k)} \subset G^{\times k}$ the subset of $k$-composable morphisms, that is 
\begin{align}
\begin{split}
    G^{(k)}&=\{(\alpha_1,\ldots,\alpha_k)\in G^{\times k}\mid t(\alpha_{i+1}) = s(\alpha_{i}), i=1,\ldots,k-1\}\\
    &=\underbrace{G\times_{(s,t)}\times\dotsm\times_{(s,t)}G}_{k}.
\end{split}
\end{align}
We denote by $m\colon G^{(k)} \to G$ the multiplication (composition of morphisms).  
\begin{defn}[Lie groupoid]
A \emph{Lie groupoid} is a groupoid where $M$ and $G$ are smooth manifolds and all structure maps are smooth. The source and target maps are surjective submersions, which guarantees that the spaces of $k$-composable morphisms are smooth manifolds.
\end{defn}
A particular item of interest are Lie groupoids with a symplectic structure \cite{Weinstein2}. %This gives rise to the following definition:
\begin{defn}[Symplectic groupoid]
\label{Symplectic groupoid}
A \emph{symplectic groupoid} is a Lie groupoid $G\rightrightarrows M$, where the space of morphisms is endowed with a symplectic form $\omega\in \Omega^2(G)$ such that the graph of the multiplication $m\colon G\times G\to G$ is a Lagrangian submanifold of $(G,\omega)\times (G,\omega)\times (G,-\omega)$.
\end{defn}
The definition above is equivalent to saying that the symplectic form $\omega$ is multiplicative, i.e.
\begin{equation}
m^*(\omega)=\pi_1^*(\omega)+\pi_2^*(\omega),    
\end{equation}

where $\pi_1$ and $\pi_2$ are projections of $G^{(2)} = G \times_{(s,t)}G$ onto its first and second component, respectively.
Definition \ref{Symplectic groupoid} is restrictive, e.g. one can show that there are no \emph{symplectic groups}. Furthermore, the following theorem holds \cite{Weinstein2}:
\begin{thm} Let $(G, \omega)\rightrightarrows M$ be a symplectic groupoid. Then 
\begin{enumerate}[$(i)$]
    \item There is a unique Poisson structure $\Pi$ on $M$ such that the source map $s$ is a Poisson map.
    \item If $\varepsilon$ denotes the unit map, then $\varepsilon(M)$ is a Lagrangian submanifold of $G$. 
    \item If $\iota$ denotes the inverse map, then the graph of $\iota$ is a Lagrangian submanifold of $G\times G$.
\end{enumerate}    
\end{thm}

\subsection{Groupoid convolution algebras}
There are several equivalent ways to define a convolution algebra on groupoids. They differ on the choice of the spaces in which the measures are defined. 
We first recall the construction of a Haar system on source fibers \cite{Connes,Higson}, and then we describe an equivalent system of measures on $(s\times t)$-fibers. The latter is more suitable for the generalization to relational groupoids. In the sequel $C_c(G)$ denotes the space of continuous functions on $G$ with compact support. %\textcolor{blue}{should we use Westman's approach instead? They are certainly equivalent}
\begin{defn}[Right Haar system on $s$-fibers]\label{Haar sfibers}
A \emph{right Haar system} on a Lie groupoid $G\rightrightarrows M$ is a smooth family of smooth measures $(\mu_x)_{x\in M}$ on the source fibers $G_x:=s^{-1}(x)$ such that 
\begin{enumerate}[$(i)$]
    \item For all $f \in C_c(G)$, then $s_*f(x) = \int_{G_x}fd\mu_x$ defines a smooth function $s_*f \in C_c(M)$.
    \item For $\gamma \colon x \to y$, the right-multiplication diffeomorphism $R_\gamma\colon G_y \to G_x$ is measure-preserving: 
    \begin{equation}
        (R_\gamma)_*\mu_y = \mu_x.
    \end{equation}
\end{enumerate}
\end{defn}

%\textcolor{magenta}{Shall we also recall again the definitions of usual convolution algebra and $C^*$-algebra?}
\begin{defn}[Groupoid convolution algebra]
Let $G\rightrightarrows M$ be a Lie groupoid with a right Haar system $(\mu_x)_{x\in M}$. Then its \emph{groupoid convolution algebra} is $(C_c(G,\C),\star)$, continuous functions with compact support on $G$ with values in $\mathbb C$, equipped with the groupoid convolution product \[\star\colon C_c(G,\C)\times C_c(G,\C)\to C_c(G,\C)\] defined by
\begin{equation}
\label{eq:convolution_product}
    (f_1\star f_2)(\gamma) = \int_{G_{s(\gamma)}}f_1(\gamma \circ \eta^{-1})f_2(\eta)d\mu_{s(\gamma)}(\eta).
\end{equation}
\end{defn}
\begin{prop}
The convolution product $\star$, defined as in \eqref{eq:convolution_product}, is associative.
\end{prop}
\begin{proof}
Let $f_1,f_2,f_3\in C_c(G,\mathbb C)$ and consider, on the one hand, 
   \begin{align*}
   ((f_1\star f_2)\star f_3)(\gamma) &= \int_{G_{s(\gamma)}}(f_1\star f_2)(\gamma \circ \eta^{-1})f_3(\eta)d\mu_{s(\gamma)}(\eta) \\
  &= \int_{G_{s(\gamma)}}\int_{G_{s(\gamma\circ\eta^{-1})}}f_1(\gamma \circ \eta^{-1}\circ\beta^{-1})f_2(\beta)\; d\mu_{s(\gamma)}(\beta)\; f_3(\eta)d\mu_{s(\gamma)}(\eta)
\end{align*}
Now set $\tau = \beta \circ \eta$. We have $R_\eta\colon G_{s(\gamma \circ \eta^{-1})} \to G_{s(\gamma)}$, and using right-invariance of the measure, it follows that the above expression equals 
\begin{align}
    &\int_{G_{s(\gamma)}}\int_{G_{s(\gamma)}}f_1(\gamma \circ \eta^{-1}\circ(\tau \circ \eta^{-1})^{-1})f_2(\tau \circ \eta^{-1})\; d\mu_{s(\gamma \circ \eta^{-1})}(\beta)\; f_3(\eta)d\mu_{s(\gamma)}(\eta) \notag \\
    &=   \int_{G_{s(\gamma)}}\int_{G_{s(\gamma)}}f_1(\gamma \circ \tau^{-1})f_2(\tau \circ \eta^{-1})\; d\mu_{s(\gamma)}(\tau)\; f_3(\eta)d\mu_{s(\gamma)}(\eta) \label{eq:associativeI}
\end{align}
On the other hand, 
\begin{align*}
(f_1\star (f_2\star f_3))(\gamma) &= \int_{G_{s(\gamma)}}f_1(\gamma \circ \eta^{-1})(f_2\star f_3)(\eta) \; d\mu_{s(\gamma)}(\eta) \\
  &= \int_{G_{s(\gamma)}}\int_{G_{s(\eta)}}f_1(\gamma \circ \eta^{-1})f_2(\eta \circ \tau^{-1})f_3(\tau)d\mu_{s(\eta)}(\tau)\;d\mu_{s(\gamma)}(\eta)
\end{align*}
This expression equals \eqref{eq:associativeI} upon exchanging $\eta$ and $\tau$. 
\end{proof}
The following equivalent definition of groupoid Haar system can be found in \cite{Westman}.
\begin{defn}[Haar systems on $(s,t)$-fibers] Let $G_{(x,y)}=\{g \in G\mid s(g)=x, t(g)=y\}$.
A \emph{Haar system} on these fibers is defined similarly as in Definition \ref{Haar sfibers}, with the requirement that, for $f\in C_c(G_{(x,y)},\mathbb C)$, if 
\[\mu_{xy}(f)=\int_{G} f(g_{x,y})d\mu(g_{x,y}),\]
the function 
\begin{eqnarray*}
\mu(f):M \times M&\to& \mathbb C\\
(x,y)&\mapsto& \mu_{x,y}(f\vert_{G_{x,y}})
\end{eqnarray*}
is in $C_c(M\times M)$, whenever $f\in C_c(G)$. 
\end{defn}
\begin{ex}[Action groupoid]
Let $G$ be a locally compact group acting continuously on a locally compact Hausdorff space $X$, then $G\times X$   (as an \emph{action groupoid}) admits a
(right) Haar system $\{\delta_x, \mu \}$ where $\mu$ is a Haar measure on $G$ and $\delta_x$ is the Dirac measure at $x\in X$.
\end{ex}
\begin{rem}
The groupoid convolution algebra has an involutive $*$-operation given by
\begin{equation}
    f^*(\gamma) = \overline{f(\gamma^{-1})}. 
\end{equation}
In order to obtain groupoid $C^{*}$-algebras we need to use completion with respect to a certain norm and a given convolution algebra representation.  
\end{rem}
\begin{defn}[Left regular representation]
The \emph{left regular representation} of the groupoid convolution algebra is a map, for all $x\in M$,
\begin{eqnarray*}
\lambda_x: C_c(G) \to \mathcal B(L^2(G_x))
\end{eqnarray*}
which for $f \in C_c(G), h\in L^2(G_x)$ and $\gamma \in G_x$ is given by
\[(\lambda_x(f)h)(\gamma)= (f\star h) (\gamma)=\int_{G_s(\gamma)} f(\gamma \circ \eta^{-1})h(\eta)d\mu_{s(\gamma)}(\eta).\]
\end{defn}

\begin{defn}[Reduced groupoid $C^*$-algebra]
The \emph{reduced} groupoid $C^{*}$-algebra of $G$ is the completion of a groupoid convolution algebra $C_c(G)$ with respect to the norm 
\[\vert\vert f \vert\vert=\sup_{x}\vert \vert \lambda_x(f) \vert \vert_{\mathcal B(L^2(G_x))}.\]
\end{defn}

%However, it is not a $C^*$-algebra, but one can complete it (via $C^*$-completion) to one. %{\color{magenta} Maybe we should explain why it is not a $C^*$-algebra? How does the $C^*$-property ($\|f^*\star f\|=f^2$) work? What norm do we take?}
%{\color{blue} In the group case, the norm is defined via a universal property. It seems that a similar thing van be done for groupoids. We can focus on the convolution algebra, and leave the discussion on $C^*$-algebra for the next paper.}
%\begin{defn}[Groupoid $C^*$-algebra]
%{\color{magenta} I guess that here we need the $C^*$-property somehow... }

%\begin{ex}
%{\color{blue} Perhaps we can add an explicit example of a groupoid $C^*$-algebra. The $C^*$-algebra for a pair groupoid is a good example}
%\end{ex}

\section{Relational groupoids} 
\label{sec:Relational_groupoids}
\subsection{The category of relational groupoids}

\begin{defn}[Relational groupoid]
\label{defn:relational_groupoid}
A \emph{relational groupoid} is a triple $(\mathcal{G},L,I)$ such that
\begin{enumerate}
    \item $\mathcal G$ is a set. 
    %{\color{magenta} Is it just a set?}
    \item $L$ is a subset of $\mathcal G\times \mathcal G \times \mathcal G$
    \item $I: \mathcal G \to \mathcal G$ is a function,
\end{enumerate}
satisfying the following axioms:
\begin{itemize}
 \item \textbf{\underline{A.1}} $L$ is cyclically symmetric, i.e. if $(g,h,k) \in L$, then $(h,k,g) \in L$.

\item \textbf{\underline{A.2}} $I$ is an involution (i.e. $I^2=\id$).\\

\item \textbf{\underline{A.3}}
Let $T$ denote the transposition map:
\begin{eqnarray*} 
T\colon  \mathcal G \times \mathcal G &\to& \mathcal G \times \mathcal G\\
(g,h) &\mapsto& (h,g).
\end{eqnarray*} Then
\begin{equation}
\label{eq:A.3}
I \circ L= L \circ T\circ (I \times I). 
\end{equation}
%{\color{magenta} We should perhaps say what is the subscript $rel$ and what does ``overline'' mean}

%, both sides of the equaliy correspond to immersed Lagrangian submanifolds.\\

\item \textbf{\underline{A.4}} 
Let \[L_3:= I \circ L\colon  \mathcal G \times \mathcal G \nrightarrow \mathcal G.\] 
%{\color{magenta} Maybe we can also say what $\nrightarrow$ means}
The following equality holds \begin{equation}\label{asso} 
L_3 \circ (L_3 \times \id)= L_3\circ (\id \times L_3)\end{equation}
%is a morphism $\colon  \mathcal G ^3 \nrightarrow  \mathcal G$. \\

%\begin{remark}\emph{The part 2 of A.4. follows automatically in the finite dimensional case from the fact that, since $I$ is an antisymplectomorphism, its graph is Lagrangian, 
%therefore $L_3$ is Lagrangian, and so $(\Id \times L_3)$ and $(L_3 \times \Id).$}
%\end{remark}
\item \textbf{\underline{A.5}} 
Denoting by $L_1$ the morphism $L_1:= L_3\circ I \colon * \nrightarrow \mathcal G$, then
\begin{equation} \label{unit}
 L_3\circ(L_1 \times L_1)= L_1.
\end{equation}

\item \textbf{\underline{A.6}}

If we define the morphism $$L_2:=L_3\circ (L_1 \times \id)\colon  \mathcal G \nrightarrow \mathcal G,$$
then the following equations hold
\begin{enumerate}[$(i)$]
\item
\begin{equation}
\label{eq:A6.1}
L_2=L_3\circ (\id \times L_1).
\end{equation}
\item $L_2$ leaves $L_1, L_2$ and $L_3$ invariant, i.e.
\begin{eqnarray} 
L_2\circ L_1&=& L_1\label{invariance1}\\
L_2\circ L_2&=& L_2\label{invariance2}\\
L_2\circ L_3&=& L_3\circ (L_2 \times L_2)=L_3\label{invariance3}.
\end{eqnarray}
\item 
\begin{equation}
I\circ L_2= L_2\circ I\mbox{ and } L_2 \mbox{ is a symmetric relation.}\label{symmetric}
\end{equation}
\end{enumerate}
\end{itemize}
\end{defn}

\begin{defn}[Relational Lie groupoid] A \emph{relational Lie groupoid} is a relational groupoid $(\mathcal G,L,I)$ such that $\mathcal G,L$ and $I$ are smooth manifolds and smooth relations, respectively.

\end{defn}

The next proposition says we can equally well define a relational groupoid through the relations $L_1,L_2,L_3$ defined in Definition \ref{defn:relational_groupoid} above.

%The next proposition says we can equally well define a relational groupoid through the relations $L_1,L_2,L_3$, defined as follows:
%\begin{defn}[Structure relations]
%The \emph{structure relations} $L_1,L_2$ and $L_3$ are defined by:
%\begin{enumerate}
%    \item (\emph{Unit Relation}) $L_1:=L\circ \mathrm{Graph}(I): \ast \nrightarrow \mathcal G$.
%    \item (\emph{Equivalence Relation})$L_2:= L \circ (L_1 \times \id): \mathcal G \nrightarrow \mathcal G$.
%    \item (\emph{Multiplication Relation}) $L_3:= I \circ L: \mathcal G \times \mathcal G \nrightarrow \mathcal G$.
%\end{enumerate}
%{\color{magenta} Should $L$ here be $L_{rel}$?}
%\end{defn}
\begin{prop}
The data $(\mathcal{G},L,I)$ and $(\mathcal{G}, I, L_1,L_2,L_3)$ are equivalent.
\end{prop}
\begin{proof}
Clearly we are able to obtain the relations $L_i$ from $(\mathcal{G},L,I)$. Now, assume that $(\mathcal{G}, I, L_1,L_2,L_3)$ is given. Then $L$ is recovered using that $L=I \circ L_3$ %{\color{magenta} Again, shouldn't it be $L_{rel}$, $I_{rel}$?}
,and therefore the axioms \textbf{\underline{A.1}}--\textbf{\underline{A.6}} can be written just in terms of $L_i$ and $I$. 
\end{proof}
\begin{defn}[Relational subgroupoid] Let $(\mathcal G, L, I)$ be a relational groupoid.
A relational groupoid $(\mathcal H, L_{\mathcal H}, I_{\mathcal H})$ is a \emph{relational subgroupoid} of ${\mathcal G, L, I}$ if $\mathcal H \subseteq \mathcal G,\,L_{\mathcal H}\subseteq L, \, I_{\mathcal H}\subseteq I$.
\end{defn}
\begin{defn}[Morphism of relational groupoids]
Let $(\mathcal{G}_1,L_1,I_1)$ and $(\mathcal{G}_2,L_2,I_2)$ be two relational groupoids. A morphism $F: \mathcal{G}_1\to \mathcal{G}_2$ is a relational subgroupoid of $\mathcal{G}_1 \times \mathcal{G}_2$.
\end{defn}
We can extend the category of groupoids $\bf Grpd$ to the \emph{category} of \emph{relational groupoids} $\bf RelGrpd$. %\marginpar{Maybe it is possible to describe this as a categoroid as for ${\bf Symp}^{Ext}$ in the RSG paper.}

%\marginpar{Put this remark after the definition? (Of objects and morphisms)}
\begin{rem}
The category $\bf RelGrpd$ is endowed with an involution $$\dagger\colon ({\bf RelGrpd})^{op}\to {\bf RelGrpd}$$ that is the identity on objects and is the relational converse of morphisms, i.e. for $f\colon A\nrightarrow B$ we get $f^\dagger:=\{(b,a)\in B\times A\mid (a,b)\in f\}$. 
\end{rem}
%{\color{magenta} Any morphism? $\mathcal{G}_i$ are sets by definition, so I don't really understand what the morphism structure should be. Is it just any map between these two sets? Maybe we can also give an example?}
%{\color{blue} Add the reduction relation as an example of a morphism}

\subsection{Graphical interpretation of the axioms}
\begin{itemize}
\item The cyclicity axiom  \textbf{\underline{A.1}} encodes the cyclic behaviour of the multiplication and inversion maps for groups, namely, if $g,h,k$ are elements of a group $G$ with unit $e$ such that $ghk=e$, then $gh=k^{-1}, \, hk= g^{-1}, kg=h^ {-1}$.
\item \textbf{\underline{A.2}} encodes the involutivity property of the inversion map of a group, i.e. $(g^{-1})^{-1}=g, \forall g \in G$.
\item \textbf{\underline{A.3}} encodes the compatibility between multiplication and inversion:
$$(gh)^{-1}=h^{-1}g^{-1}, \forall g,h \in G.$$
\item \textbf{\underline{A.4}} encodes the associativity of the product: $g(hk)=(gh)k, \forall g,h,k \in G$.
\item \textbf{\underline{A.5}} encodes the property of the unit of a group of being idempotent: $ee=e$.
\item The axiom \textbf{\underline{A.6}}  states an important difference between the construction of relational  groupoids and usual groupoids. The compatibility between the multiplication and the unit is defined up to an equivalence relation, denoted by $L_2$, whereas for groupoids such compatibility is strict; more precisely, for groupoids such equivalence relation is the identity. In addition, the multiplication and the unit are equivalent with respect to $L_2$.
\end{itemize}
Figure \ref{fig:Ax} illustrates the diagrammatics of the relational groupoid axioms.

\begin{figure}[h]
\centering
\def\svgwidth{360pt}
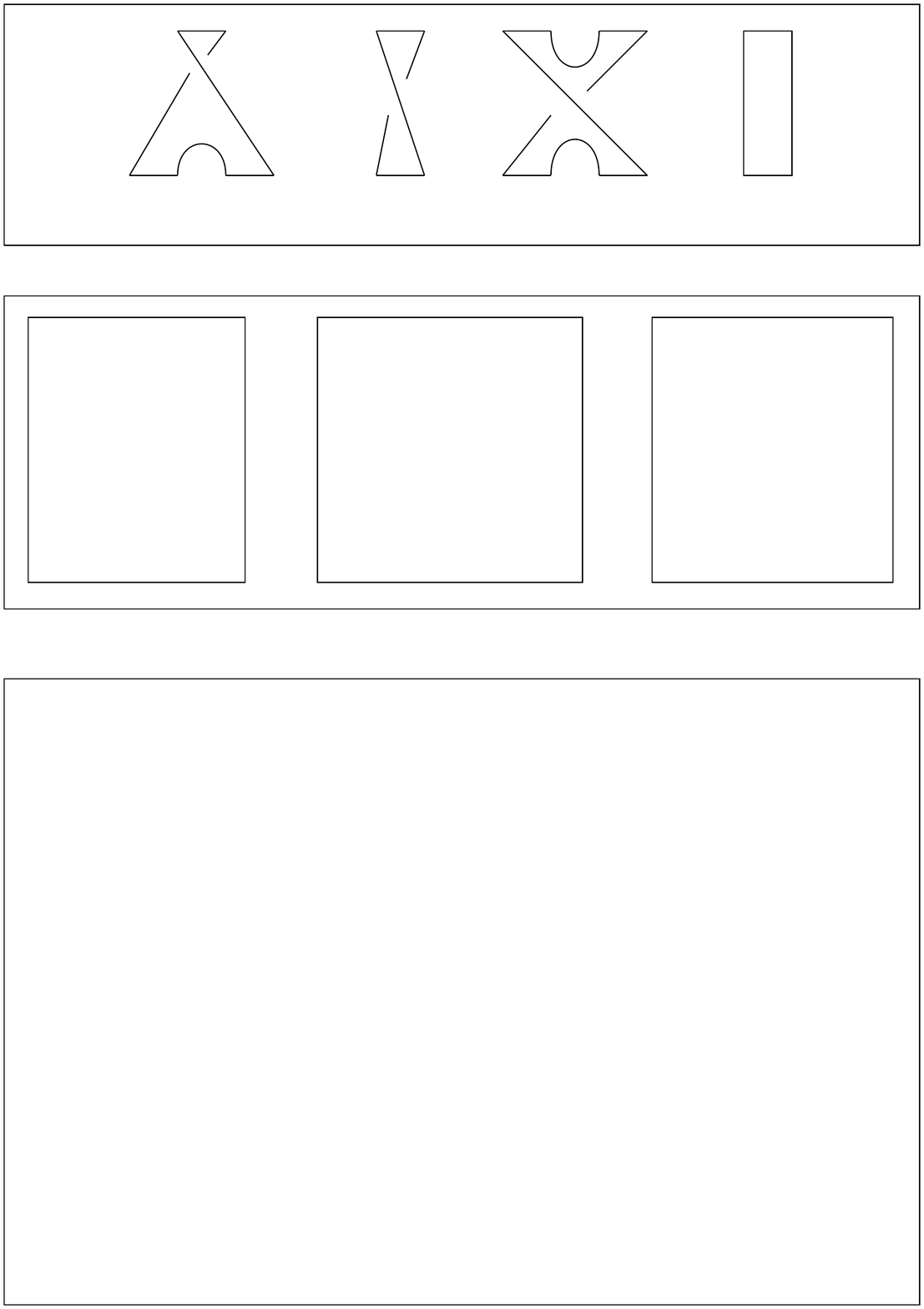
\caption{Diagrammatics of a relational groupoid}
\label{fig:Ax}
\end{figure}

\subsection{Relational groups as relational groupoids}
The following is a representative example of relational groupoids: it is given by a group $\mathcal G$ and a normal subgroup $\mathcal H$ of $\mathcal G$.
\begin{ex}[Relational groups]\label{RelGroup}
Let $\mathcal G$ be a group with multiplication $m\colon \calG\times \calG \to \calG$ and $\calH\triangleleft \calG$ a normal subgroup. Denote by $\sim_\calH \subset \calG \times \calG$ the equivalence relation $ g_1 \sim_\calH g_2 \Leftrightarrow \exists h \in \calH$ such that $m(h,g_1) = g_2.$ Moreover, define 
\begin{align}
    L_3 &:= \{(g_1,g_2, m(m(g_1,g_2),h))\, \vert\, g_1,g_2\in \calG, h\in \calH \} \subset \calG \times \calG \times \calG, \\
    L_2 &:= \sim_\calH \subset \calG \times \calG, \\
    L_1 &:= \calH \triangleleft \calG, \\
    I &:= g \mapsto g^{-1}.
\end{align}
Then the quintuple $(\calG,I,L_1,L_2,L_3)$ defines a relational groupoid. 
\end{ex}

\begin{proof}
We need to check the axioms of a relational groupoid as in Definition \ref{defn:relational_groupoid} explicitly. We will not always write the multiplication map by $m$ but instead $g_1g_2:=m(g_1,g_2)$. To see Axiom \textbf{\underline{A.1}}, let $(g_1,g_2,g_3) \in L$, we want to show $(g_3,g_1,g_2) \in L$ or $(g_3,g_1,g_2^{-1})\in L_3$ Then $g_3 = (g_1g_2h)^{-1}$ for some $h \in \calH$ and $g_3g_1 = h^{-1}g_2^{-1}$. By normality, $h^{-1}g_2^{-1} = g_2^{-1}h'$ with $h' \in \calH$ and so $(g_3,g_1,g_2^{-1})\in L_3$. Clearly $I$ is an involution, hence \textbf{\underline{A.2}} holds. Axiom \textbf{\underline{A.3}} \eqref{eq:A.3} follows from the fact that $L=I\circ L_3$ and that for group elements $g_1,g_2$ we have $I(g_1g_2)=(g_1g_2)^{-1}=g_2^{-1}g_1^{-1}=I(g_2)I(g_1)$, and normality.
Next we want to show \textbf{\underline{A.4}} \eqref{asso}. Consider an element $(g_1,g_2,g_1g_2h)\in L_3$ for some $h\in \calH$. Consider the diagram of relations
\begin{equation}
% \calG\times \calG\times \calG \arrow{r}[yshift=1ex]{L_3\times \id}[swap,yshift=1.5ex]{/} \arrow[bend right=50]{rr}[yshift=1ex]{L_3\circ(L_3\times \id)}[swap,yshift=1.5ex]{/} \arrow[bend left=50]{rr}[yshift=1ex]{L_3\circ(\id\times L_3)}[swap,yshift=1.5ex]{/}& \calG\times \calG \arrow{r}[yshift=1ex]{L_3}[swap,yshift=1.5ex]{/} & \calG
\begin{tikzcd}
                                                                                                                                                                                                                                                                                               & \mathcal{G}\times\mathcal{G} \arrow[rd, "L_3", bend left=49,negate]   &             \\
\mathcal{G}\times\mathcal{G}\times\mathcal{G} \arrow[ru, "L_3\times\mathrm{id}", bend left=49,negate] \arrow[rd, "\mathrm{id}\times L_3"', bend right=49,negate] \arrow[bend right,negate]{rr}[yshift=1ex]{L_3\circ(\mathrm{id}\times L_3)} \arrow[bend left,negate]{rr}[yshift=-4ex]{L_3\circ(L_3\times\mathrm{id})} &                                                                & \mathcal{G} \\
                                                                                                                                                                                                                                                                                               & \mathcal{G}\times\mathcal{G} \arrow[ru, "L_3"', bend right=49,negate] &            
\end{tikzcd}
\end{equation}\\
% \marginpar{I'm sure what is meant by this diagram - the middle is the same trivially as the lower arrow, but what about the upper one? Wouldn't this rather be a ``commutative square'' situation? \textcolor{blue}{We could omit this diagram, or add the commutative square, for clarity.}}

Let us first look at the relation $L_3\circ (L_3\times \id)$. It follows that  $(g_1,g_2,g_3)\sim (g_1g_2h_1,g_3)\sim g_1g_2h_1g_3h_2$ with $h_1,h_2\in \calH$ and $g_1,g_2,g_3 \in \mathcal G$. Now using normality of $\calH$ we get $h_1g_3 = g_3h_1'$, for some $h_1'\in\calH$, and setting $\bar h:=h_1'h_2$  we get $g_1g_2h_1g_3h_2 = g_1g_2g_3\bar h\in g_1g_2g_3\calH$. 
% \marginpar{Not sure I quite understand this. We don't necessarily have $g_1g_2h_1g_3h_2 = g_1g_2g_3h_1h_2$ (which seems to be implied here). But, it is true that $h_1g_3 = g_3h_1'$ for some $h_1' in H$, and then $g_1g_2g_3h_1'h_2 \in g_1g_2g_3\calH$. \textcolor{blue}{Maybe we can rewrite this argument}} 
On the other hand, if we look at the relation $L_3\circ (\id \times L_3)$, we get $(g_1,g_2,g_3)\sim (g_1,g_2g_3h_3)\sim g_1g_2g_3h_3h_4 %$ and setting $\tilde h:=h_3h_4$, we get $g_1g_2g_3\tilde h
\in g_1g_2g_3\calH$. 
% \marginpar{So both relations are the multiplication map composed with the relation $\tilde_H$. } 
Next we show \textbf{\underline{A.5}} \eqref{unit}. Consider the relations
\begin{equation}
\begin{tikzcd}
* \arrow[negate]{r}[yshift=1ex]{L_1 \times L_1} \arrow[bend right=70,negate]{rr}[yshift=0.8ex]{L_3\circ (L_1\times L_1)}\arrow[bend left=70,negate]{rr}[yshift=1ex]{L_1}& \calG\times \calG \arrow[negate]{r}[yshift=1ex]{L_3} & \calG
\end{tikzcd}
\end{equation}
%which sends the point to 
where we have $* \sim (h_1,h_2)$ with $h_1,h_2\in \calH$ and then $(h_1,h_2)\sim h_1h_2h_3=:\tilde h\in \calH$. On the other hand we have a relation $L_1$ for $* \sim h$ for any element $h$ of $\calH$. Next we show \textbf{\underline{A.6}} \eqref{eq:A6.1}. Let us look at 
\begin{equation}
\begin{tikzcd}
\calG\times * \arrow[negate]{r}[yshift=1ex]{\id \times L_1} \arrow[bend right=70,negate]{rr}[yshift=1ex]{L_3\circ (\id\times L_1)} \arrow[bend left=70,negate]{rr}[yshift=1ex]{L_2}& \calG\times \calG \arrow[negate]{r}[yshift=1ex]{L_3} & \calG
\end{tikzcd}
\end{equation}
and consider first the relation $L_3\circ (\id \times L_1)$.
Take $g\in \calG$ then we get $g\sim (g,h_1)$ for $h_1\in \calH$ and $(g,h_1)\sim g\bar h\in g\calH$ with $\bar h =h_1h_2\in \calH$. On the other hand if we consider the relation $L_2$ we get for $g\in \calG$ and $h\in \calH$ that $g\sim gh\in g\calH$. Next we show \textbf{\underline{A.6}} \eqref{invariance1}.
Consider the relations

\begin{equation}
\begin{tikzcd}
* \arrow[negate]{r}[yshift=1ex]{L_1} \arrow[bend right=70,negate]{rr}[yshift=1ex]{L_2\circ L_1} \arrow[bend left=70,negate]{rr}[yshift=1ex]{L_1}& \calG \arrow[negate]{r}[yshift=1ex]{L_2} & \calG
\end{tikzcd}
\end{equation}

and let us first look at the relation $L_2\circ L_1$. The relation $L_1$ is $* \sim h$ for any element in $h\in \calH$ and then by $L_2$ we get $* \sim h\sim \tilde h:=h\bar h\in \calH$. Since $\calH$ is in particular a subgroup, we get $* \sim h$ for any $h \in \calH$, which is $L_1$. 
% which is on the other hand the same as if $L_1$ maps the point to an element $\tilde h\in \calH$.
Finally, we show \textbf{\underline{A.6}} \eqref{invariance2}. Then we have the following relations

\begin{equation}
\begin{tikzcd}
\mathcal G \arrow[negate]{r}[yshift=1ex]{L_2} \arrow[bend right=70,negate]{rr}[yshift=1ex]{L_2\circ L_2} \arrow[bend left=70,negate]{rr}[yshift=1ex]{L_2}& \calG \arrow[negate]{r}[yshift=1ex]{L_2} & \calG
\end{tikzcd}
\end{equation}
and clearly we have that $L_2$ gives for an element $g\in \calG$ that $g\sim gh_1$ for $h_1\in \calH$ and then again $gh_1\sim gh_1h_2=g\bar h\in g\calH$ with $\bar h=h_1h_2 \in \calH$, which gives the same relation as just $L_2$. Similarly, we can show for \textbf{\underline{A.6}} \eqref{invariance3} that also the last diagram commutes, completing the proof.

\begin{equation}
\begin{tikzcd}
 & \calG\times \calG \arrow[bend left,negate]{dddr}[yshift=0.5ex]{L_3}& \\
 & &\\
 & &\\
\calG \times \calG \arrow[negate]{r}[yshift=1ex]{L_3} \arrow[bend right=70,negate]{rr}[yshift=1ex]{L_2\circ L_3} \arrow[bend left=70,negate]{rr}[yshift=1ex]{L_3} \arrow[bend left,negate]{uuur}[yshift=1ex]{L_2\times L_2}& \calG \arrow[negate]{r}[yshift=1ex]{L_2} & \calG
\end{tikzcd}
\end{equation}
\end{proof}
For later use we record some simple examples of relational groups below.  
\begin{ex}[A finite example] 
\label{RelGroup2}
One can check that $\calG=\mathbb{Z}_4$ with normal subgroup $\calH=\mathbb{Z}_2\triangleleft \mathbb{Z}_4$ together with the canonical relations as in Example \ref{RelGroup} is a relational group and hence a relational groupoid.
\end{ex}
\begin{ex}[Discrete example: integers modulo $n$]\label{ex:discrete}
We want to consider the example of $\Z/n\Z$ for some $n\geq2$.
Hence, let $G=\Z$, $L_1=n\Z$, $L_2=\{(a,b)\in \Z\times \Z\mid a-b\in n\Z\}$ and $L_3=\{(a,b,c)\in \Z\times \Z\times \Z\mid \exists k \in \Z\colon  a+b+nk=c\}$. 
%Note that the group operation is given by addition. 
%We get a discrete convolution product on %$C_c(\Z,\C)$ by
%\begin{equation}
%    (f_1\star f_2)(a)=\sum_{b\in \Z}f_1(a-b)\cdot f_2(b)
%\end{equation}
%\textcolor{red}{This example is wrong, it contradicts the axiom $L_2 \circ L_3 = L_3$. } \textcolor{blue}{fixed it}
\end{ex}

\begin{ex}[A continuous example]
\label{RelGroup3}
Let $\calG=S^1:=\{z\in \C\mid \vert z\vert =1\}$ and let $L_1=\{\zeta\in \C\mid \zeta^n=1\}\triangleleft S^1$ be the normal subgroup of $n$-th roots of unity.  We define the relation $$L_2=\left\{(z,w)\in S^1\times S^1\,\big|\, \exists k\in\{0,1,\ldots,n-1\}\textnormal{ such that }\arg(w)-\arg(z)=\frac{2\pi k}{n}\right\}.$$
In particular, $z\sim w$ if and only if $\exists \zeta\in L_1$ such that $z\cdot \zeta=w$ for $z,w\in S^1$.
We define $L_3=\{(z,w,z\cdot w\cdot\zeta)\in S^1\times S^1\times S^1\mid z,w\in S^1,\zeta\in L_1\}\subset S^1\times S^1\times S^1$.
Moreover, we define the map $I$ by complex conjugation $z\mapsto \bar z$. Then one can check that this is indeed a relational group as in Example \ref{RelGroup} and hence a relational groupoid.
\end{ex}

\subsection{Additional Examples}
Here are some other simple examples of relational groupoids that are not relational groups. 
\begin{ex}[Relational group bundle]
Example \ref{RelGroup} can be extended to a parametrized family of relational groups. The local model in this example is $\mathcal G \times \mathbb R^k$, where at each point $p$ in $\mathbb R^k$ there is a fiber that we identify with $\mathcal G$. For instance, Example \ref{RelGroup3} can be extended to the relational bundle  $S^1\times \mathbb R \to \mathbb R$, where each fiber is isomorphic to the relational group $\mathcal G=S^1$ and at each fiber, the normal subgroup of $n$-th roots of unity is chosen.
%The relational groupoid convolution algebra is obtained pointwise.

\end{ex}

\begin{ex}[Groupoids]
Of course, groupoids, as in Section \ref{subsec:Groupoids}, are also examples of relational groupoids. If $(G,m)$ is a groupoid, then $L:=\mathrm{Graph}(I\circ m)$, $I(g) = g^{-1}$ makes $G$ into a relational groupoid. In this case, the special relations are given as follows: $L_1 \subset G$ is the unit section, $L_2 = \mathrm{diag}_G\subset G \times G$ is the diagonal, and $L_3=\mathrm{Graph}(m)$.
\end{ex}
% {\color{blue}Is it true that a relational groupoid with $L_2=id$ ``is'' a groupoid? this would seem to follow from theorem 3.2.}

\subsection{Relational symplectic groupoids}
\begin{defn} \label{Rel}
A \emph{relational symplectic groupoid} is a  relational groupoid $(\mathcal G,\, L,\, I)$ where 
\begin{enumerate}
 \item $\mathcal G$ is a weak symplectic manifold\footnote{In the infinite-dimensional setting we restrict to the case of Banach manifolds (when the regularity type of fields is fixed) $\calG$ endowed with a closed 2-form $\omega$, such that the induced map $\omega^{\sharp}: T\mathcal G \to T^*\mathcal G$ is injective. The result also holds for smooth fields and Fr\'echet manifolds.}%{\color{magenta} Maybe we should say what this means.}
\item $L$ is an immersed Lagrangian submanifold of $\mathcal{G}\times\mathcal{G}\times \overline{\mathcal{G}}$, where $\overline {\mathcal G}$ denotes $\mathcal G$ equipped with the negative of the given symplectic form.
%{\color{magenta} Maybe we can say what the ``overline'' means.}
 \item $I$ is an antisymplectomorphism of $\mathcal G$.
\end{enumerate}
\end{defn}
\begin{ex}\label{ex:RSGPSM} In \cite{Relational, ContrerasThesis} it is proven that the phase space of the Poisson Sigma Model (PSM) \footnote{where the source space is a disk, and the target space is a Poisson manifold.} is an example of a infinite-dimensional relational symplectic groupoid.  
\end{ex}

%\subsection{The constrain set, source and target}

\subsection{Reduction of relational groupoids}
One of the key properties of a relational groupoid is the fact that $L_2$ encodes the information of an equivalence relation. In general $L_2$ is not necessarily an equivalence relation on the whole of $\mathcal G$, but on a subset $\calC$, called the \emph{constraint set}. 
\begin{prop}
\label{prop:constraint}
Define 
\[\calC:=L_2\circ \mathcal G,\]
where $\mathcal G$ is considered as the relation $\ast \nrightarrow \mathcal G$. %{\color{magenta} Isn't $\mathcal G_{rel}$ just $L_1$ by definition 3.3?}.
Then $L_2$ is an equivalence relation on $\calC$.
\end{prop}
\begin{proof}
The fact that $L_2$ is transitive and symmetric follows from Axiom \textbf{\underline{A.6}} ($L_2\circ L_2=L_2$ and $L_2^{\dagger}=L_2$). Reflexivity follows from the definition of $\calC$.
\end{proof}

\begin{thm}
Let $(\mathcal{G},I,L_1,L_2,L_3)$ be a relational groupoid. Then $\underline{\calG}=\mathcal{C}/L_2$ is a groupoid and the quotient map $q$ is a morphism of relational groupoids.
%{\color{magenta} $\mathcal C=C$? What is a morphism of relational groupoids?}
\end{thm}
\begin{proof}
First, let us consider the following relations, that are the relational analogues of the source and target maps:
\[\mathcal S:= \{(c,\ell)\in \calC\times L_1\,\vert \,  \exists g \in \mathcal G \mbox{ s.t. } (\ell,c,g) \in L_3 \} \]
and 
\[\mathcal T:= \{(c,\ell)\in \calC\times L_1\,\vert \,  \exists g \in \mathcal G \mbox{ s.t. } (c,\ell,g) \in L_3 \}. \]
It follows from the definition that $\mathcal T = I \circ \mathcal S$ and furthermore, if $M:= L_1/ L_2:$
\[s:= \underline{\mathcal S}: \underline{\calG} \to M \]
is a surjective map, where $\underline{\mathcal S}=q\circ \mathcal S$ is the reduction of $\calS$.
\end{proof}
\begin{ex}[Relational group]
Following Example \ref{RelGroup}, the reduction of a relational group  is a (set theoretical) group. If in addition we impose the condition  that $G$ is a Lie group, and  $H$ is closed subgroup, then the quotient is a Lie group.
\end{ex}
\begin{ex}[Relational pair groupoid]
Let $G$ be a set. Then one can define its \emph{pair groupoid} by $G\times G\rightrightarrows G$, where $s$ and $t$ are given by projection to the first and second factor respectively. For two elements $(g,h)$ and $(h,k)$, composition is given by $(g,h)(h,k)=(g,k)$. The inverse is defined by $(g,h)^{-1}=(h,g)$. We can define a \emph{relational pair groupoid} in a similar way. For a relational version of this example, let $\mathcal G$ be a set and an equivalence relation $R$. We then define the relational groupoid $\mathcal G \times \mathcal G$, where $L_1=\mathcal G$, $L_2=R\times R$, and $L_3$ is given by composition of relations. In case that the equivalence relation $R$ is the identity, we recover the pair groupoid after $L_2$-reduction.
\end{ex}
\begin{ex}[Relational symplectic groupoid]
If $M$ is an integrable Poisson manifold, then the reduction of the infinite-dimensional relational symplectic groupoid of Example \ref{ex:RSGPSM} is a finite-dimensional symplectic groupoid integrating the Poisson manifold $M$. See \cite{Relational, ContrerasThesis} for the details of the reduction procedure and how it coincides with the gauge reduction of the Poisson Sigma Model. 
\end{ex}
The previous results and examples allow us to connect different constructions on relational groupoids with standard notions on groupoids, via reduction. 
For instance, in the next subsection we introduce actions of relational groupoids, that are necessary to describe the compatibility of measures (relational Haar systems) and the structure relations in Section \ref{sec:Relational_convolution_algebras}.

\subsection{Relational groupoid actions}
First, following the characterization of the Haar measure on groups via right-invariance, we introduce the notion of relational right action.
Right, left, and adjoint group actions are natural examples of relational group actions, and also Haar measures (on groups and relational groups) are invariant with respect to the right relational group action. Also, the conditions (\ref{item:haarsys1}) and (\ref{item:haarsys2}) of a relational Haar system in Definition \ref{def:RelHaarSys} encode the invariance with respect to relational right actions.

\begin{defn}[Relational right action]
Let $(\mathcal{G},L,I)$ be a relational groupoid and let $Z$ be a set with an equivalence relation $L_Z$. A \emph{relational right action} of $\mathcal{G}$ on $Z$ is a relation $\rho\colon Z \times \mathcal{G} \nrightarrow Z$ such that 
\begin{enumerate} \item we have $$\rho \circ (\rho \times \mathrm{id}_\mathcal{G}) = \rho \circ (\mathrm{id}_Z \times L_3)$$ as relations $Z \times \mathcal{G} \times \mathcal{G} \nrightarrow Z$, 
\item The relation $ \rho_{L_1}\colon Z \nrightarrow Z$ given by $$\rho_{L_1} = \{(x,y) \in Z \times Z\mid \exists g \in L_1, (x,g,y) \in \rho\}$$  coincides with the equivalence relation $L_Z$ on $Z$, i.e. $\rho_{L_1} = L_Z$. 
\end{enumerate}
We also say that $(\mathcal{G},L,I)$ acts on $(Z,L_Z)$ by relations from the right. Sometimes we drop $L_Z$ notation. 
\end{defn}
An obvious example is the action of a relational groupoid on itself from the right:
\begin{ex}
Let $\mathcal{G}$ be a relational groupoid. Then setting $\rho = L_3$ defines a relational right action of $\mathcal{G}$ on $(\mathcal{G},L_1)$.
\end{ex}
\begin{proof}
This follows directly from axioms {\bf \ul{A.4}} (associativity) and {\bf \ul{A.6}} (unitality) of Definition \ref{defn:relational_groupoid}. 
\end{proof}
Of course, there is an analogous definition of left relational action, and relational groupoids also act on themselves from the left. \\
We can generalize the relation $\rho_{L_1}$ defined above to an arbitrary subset of $\mathcal{G}$. 
\begin{defn}[Relational action]
Let $(\mathcal G,L,I)$ be a relational groupoid and suppose that $\mathcal{G}$ acts on $Z$ by relations and let $S\subset \mathcal G$ be any subset of $\mathcal{G}$. Then we define the relation $R_S\colon Z \nrightarrow Z$ by 
\begin{equation}
    R_S :=  \{(z_1,z_2) \in Z \times Z\mid \exists g \in S, (z_1,g,z_2) \in \rho\}
\end{equation}
For $S = \{g\}\subset \mathcal G$, we write $R_S = R_g$.
\end{defn}
For $g,h \in \mathcal{G}$, we denote $gh:=\{x \in \mathcal{G}\mid (g,h,x) \in L_3\}$. With this notation, we have the following proposition:
\begin{prop}
Let $(\mathcal G,L,I)$ be a relational groupoid and let $g,h \in \mathcal{G}$, which acts on a set $Z$ from the right. Then 
\begin{equation}
    R_h \circ R_g = R_{gh}.
\end{equation}
\end{prop}
%{\color{blue} I think it makes sense to add the diagrammatics of what a relational groupoid action means. Add an example of a relational group action. For example, the right action for relational groups}

\section{Relational convolution algebras}
\label{sec:Relational_convolution_algebras}
\subsection{Relational Haar systems}\label{ssec: Relational_Haar}
Let $(\mathcal G,L,I)$ be a relational groupoid. As for a usual groupoid, we denote by
\begin{equation}
\mathcal G^{(2)}:= \{(g,h) \in \mathcal G\times \mathcal G\mid \exists k \in \mathcal G \text{ s.t. } (g,h,k) \in L_3\}
\end{equation}
the set of composable pairs in $\mathcal G$. For $k \in \mathcal G$, we denote by $\mathcal G^{(2)}_k$ the set of pairs that compose to $k$, i.e. we have 
\begin{equation}
\mathcal G^{(2)}_k:= \{(g,h) \in \mathcal G\times \mathcal G \mid (g,h,k) \in L_3\}
\end{equation}
Recall that we have the quotient groupoid $\underline{\calG}:=\calC/L_2$ and there is a relation $q\colon \mathcal{G} \nrightarrow \underline{\calG}$ which, restricted to $\calC\times \underline{\calG}$, is the graph of a surjective map that we will also denote $q$. It is clear from the definitions that for every $g \in \calC$ we have $\mathcal{G}^{(2)}_g/(L_2 \times L_2) = \underline{\calG}^{(2)}_{q(g)}$, and we denote\footnote{We are slightly abusing notation here, $q_g$ is actually the restriction of $q\times q$ to $\mathcal{G}^{(2)}_g$.} the quotient map $q_g\colon \mathcal{G}^{(2)}_g\to \underline{\calG}^{(2)}_{\ul{g}}$, where we denote $q(g) =: \underline{g}$. \\

We will need the following terminology from measure theory (see e.g. \cite{Ambrosio+,ChangPollard}). 
\begin{defn}[Disintegrating measure]\label{def:Disintegration}
Let $\mu$ be a measure on a set $Y$ and let $q\colon Y \to X$ be a map. Moreover, let $\nu=q_*\mu$ be the pushforward measure on $X$. We say that \emph{ $\mu$ disintegrates  with respect to $q$} if there exists a family of probability measures $(\mu_x)_{x\in X}$   such that
\begin{itemize}
    \item For all $\mu$-measurable sets $E\subset Y$,  the function $x \to \mu_x(E)$ is $\nu$-measurable.
    \item $\mu_x(Y\setminus q^{-1}(x)) = 0$ for $\nu$-almost every $x \in X$.
    \item For all $\mu$-measurable functions $f\colon Y\to \R$, we have 
    $$\int_Y f d\mu = \int_X \int_{q^{-1}(x)}f(y)d\mu_x(y)d\nu(x).$$
\end{itemize}\end{defn}
A measure $\mu$ disintegrates under fairly general assumptions, e.g. when $X,Y$ are Radon spaces and $q$ is Borel-measurable. The family $\mu_x$ is uniquely determined $\nu$-almost everywhere (and in turn determines the measure $\mu$). See \cite{Disintegration} for a detailed account. \\
We now define a relational Haar system as follows. 
\begin{defn}[Relational Haar system]\label{def:RelHaarSys}
Let $\mathcal{G}$ be a relational groupoid, $\underline{\calG}$ its quotient groupoid and $q\colon \calC \to \underline{\calG}$ the quotient map. A \emph{relational Haar system} on $\mathcal{G}$ is a system of measures $\mu_g$ on $\mathcal{G}^{(2)}_g$, $g \in \mathcal{G}$ such that $\mu_g = 0$ for $g \in \mathcal{G}\setminus\calC$ and for $g \in \calC$ we have
\begin{enumerate}[$(i)$]
    \item\label{item:haarsys1} If $q(g)= q(g')$ then $(q_g)_*\mu_g = (q_{g'})_*\mu_{g'}$.
    \item\label{item:haarsys2} The system of measures $\nu_{\underline{g}}:=(q_g)_*\mu_g$ on $\underline{\calG}^{(2)}_{\underline{g}}$ defines a right Haar system on the quotient groupoid $\underline{\calG}$ (recall that $\ul{g} = q(g)$).
    \item\label{item:haarsys3} $\mu_g$ disintegrates with respect to $q_g$. 
\end{enumerate}
\end{defn}
Equivalently, a relational Haar system on $\mathcal{G}$ is determined by a right Haar system $(\nu_{\underline{g}})_{\underline{g}\in \underline{\calG}}$ on the quotient groupoid $\underline{\calG}$ and a family of probability measures $(\mu_g)_{\underline{g}_1\underline{g}_2}$ on the fibers $q_g^{-1}(\underline{g}_1,\underline{g}_2)$ for $\ul{g}_1\ul{g}_2 = \ul g$. 
\begin{rem}
A relational Haar system is invariant under the relational right action of the relational groupoid on itself, in the following sense. Let $A \subset \calG^{(2)}_g$ be an $L_2$-saturated set, i.e. ($L_2 \times L_2) \circ A = A$. Then, if we have $(g,h,k) \in L_3$, we have $\mu_g(A) = \mu_{k}((\id \times R_h) \circ A)$. 
\end{rem}

Notice also that the condition that $\mu_g$ vanishes for $g \notin \calC$ is automatic because in this case $\mathcal{G}^{(2)}_g = \emptyset$ (see also Proposition \ref{prop:ginC}). 
In some sense, the axioms presented above are the weakest possible set of axioms that ensure existence of a well-defined Haar system on the quotient. However, these measures can have extra properties with regard to the structure relations that define a relational groupoid:
\begin{defn}[$L_2$-invariant measure]\label{def:l2inv}
We say that a relational Haar system $\mu_g$ is \emph{ $L_2$-invariant} if $\mu_g = \mu_h$ whenever $(g,h) \in L_2$. 
\end{defn}
Notice that this condition is stronger than condition (i) of Definition \ref{def:RelHaarSys}, which merely demands that the pushforwards be the same. 
\begin{defn}[Split relational Haar system]\label{def:split}
We say that a relational Haar system $\mu_g$ \emph{splits} if, for all $g \in \calC$, there is a family of probability measures $(\tau^g_{\ul{g}_1})_{\ul{g}_1 \in \underline{\calG}}$ on $\mathcal{G}$ with $\tau^g_{\ul{g}_1}$ supported on $q^{-1}(\ul{g}_1)$ such that for all $\ul{g}_1,\ul{g}_2 \in \underline{\calG}$ with $\ul{g}_1\ul{g}_2 = q(g)$
    \begin{equation}
    (\mu_g)_{\ul{g}_1\ul{g}_2} = \tau_{\ul{g}_1}^g \times \tau_{\ul{g}_2}^g.
    \end{equation}
\end{defn}
\begin{defn}[Strongly split relational Haar system]\label{def:strongsplit}
We say that a relational Haar system $\mu_g$ \emph{splits strongly} if there is a family of probability measures $(\tau_{\ul{g}_1})_{\ul{g}_1 \in \ul\calG}$ on $\mathcal{G}$ with $\tau_{\ul{g}_1}$ supported on $q^{-1}(\ul{g}_1)$ such that for all $g \in \calG$ and $\ul{g}_1,\ul{g}_2 \in \underline{\calG}$ with $\ul{g}_1\ul{g}_2 = q(g)$
    \begin{equation}
    (\mu_g)_{\ul{g}_1\ul{g}_2} = \tau_{\ul{g}_1} \times \tau_{\ul{g}_2}.
    \end{equation}
\end{defn}
In particular, a strongly split relational Haar system is split and $L_2$-invariant.
\begin{ex}%\marginpar{Konstantin: clean up this example}
Consider the relational group $\calG=\mathbb{Z}_4$ with normal subgroup $\calH=\mathbb{Z}_2\triangleleft \mathbb{Z}_4$, as in Example \ref{RelGroup2}. The sets $\calG^{(2)}_g$ are given by
\begin{align*}
    \calG^{(2)}_0 &= \{(0,0),(1,1),(2,2),(3,3),(0,2),(2,0),(1,3),(3,1)\} = \calG^{(2)}_2 \\
     \calG^{(2)}_1 &= \{(1,0),(0,1),(1,2),(2,1),(3,0),(0,3),(2,3),(3,2)\} = \calG^{(2)}_3 
\end{align*}
and in the quotient we have 
\begin{align*}
    \underline{\calG}^{(2)}_{\underline{0}} &= \{(\ul{0},\ul{0}),(\ul{1},\ul{1})\}\\
     \calG^{(2)}_{\ul{1}} &= \{(\ul{1},\ul{0}),(\ul{0},\ul{1})\}
\end{align*}
Let $\mu_\ct$ denote the counting measure, $\mu_\ct(A) = \#A$. The unique Haar system (up to normalization) on the quotient is given by letting $\nu_{\ul{0}} = \nu_{\ul{1}} = \frac12\mu_\ct$, and of course this corresponds to the natural Haar measure $\mu_\ct$ on $\mathbb{Z}_2$.
%\textcolor{magenta}{Maybe we can put the $c$ as a superscript for the counting measure, since we later also talk about the Dirac measure at some $g$ which we denote $\delta_g$.} 
A relational Haar system can now be given by assigning, for $\ul{g}_1\ul{g}_2 = q(g) \in \ul{\calG}$, probability measures $(\mu_g)_{\ul{g}_1\ul{g}_2}$ on the fibers $q^{-1}_g(\ul{g}_1,\ul{g}_2) = \{\ul{g}_1,\ul{g}_1+2\} \times \{\ul{g}_2,\ul{g}_2+2\}$ for $\ul{g}_1\ul{g}_2= q(g)$. An obvious choice is to assign $(\mu_g)_{\ul{g}_1\ul{g}_2} = \frac14\mu_\ct$. This yields the family of measures $\left(\mu_g=\frac{1}{8}\mu_\ct\right)_{g\in\mathbb{Z}_4}$ %\textcolor{magenta}{$\leftarrow$ I don't understand this notation}
which is strongly split, with $\tau_{\ul{0}} = \tau_{\ul{1}} = \frac12\mu_\ct$. But other choices are possible: For instance, we can define a split, but not $L_2$-invariant measure by setting $\tau^g_{\ul{g}_1} = \delta_g$--the Dirac measure at $g \in q^{-1}(\ul{g}_1)$ - whenever $q(g) = \ul{g}_1$ and $\tau^g_{\ul{g}_1} = \frac12\mu_\ct$ otherwise. Concretely, we have 
$$\tau^0_{\ul{0}} = \delta_0, \qquad \tau^0_{\ul{1}} = \frac12\mu_\ct, \qquad \tau^2_{\ul{0}} = \delta_2, \qquad \tau^2_{\ul{1}} = \frac12\mu_\ct$$
and similarly for $\tau^1_{\ul{g}_1}$ and $\tau^3_{\ul{g}_1}$.
On the other hand, we can define an $L_2$-invariant system which is not split by letting $(\mu_0)_{\ul{g}_1\ul{g}_2} = (\mu_2)_{\ul{g}_1\ul{g}_2}$ any probability measure on $q^{-1}(\ul{g}_1,\ul{g}_2)$ which is not a product measure, for instance 
$(\mu_0)_{\ul{0}\ul{0}}(0,0) = (\mu_0)_{\ul{0}\ul{0}}(2,0) = (\mu_0)_{\ul{0}\ul{0}}(0,2) = \frac13, (\mu_0)_{\ul{0}\ul{0}}(2,2) = 0$, and similarly for the other probability measures. 
%\textcolor{magenta}{Again, here we should be careful with the notation for elements in the quotient (put $\underline{g}_1,\underline{g_2},...$ instead of $x,y,...$)}
\end{ex}
\begin{ex}
Consider again the example of the relational group $\mathcal G$ corresponding to $n\Z \triangleleft \Z$ as in Example \ref{ex:discrete}. In this case, assigning for $l\in \Z$ the counting measure $\mu_\ct$ on $\calG^{(2)}_l$ does not provide an example of relational Haar system, since it does not disintegrate in the sense of Definition \ref{def:Disintegration}. The point is that the fibers are infinite and thus the ``measures along the fibers'' $(\mu_l)_{\underline{l}_1\underline{l}_2}$ are not probability measures. 
One possibility to define a (strongly split) relational Haar system is to define probability measures $\tau_{\underline{l}}$ on $l + n\Z$, for $\underline{l} \in \Z_n$ (for instance the Dirac measure supported at some representative). However, these measures will necessarily fail to be translation-invariant.
\end{ex}
\subsection{The relational convolution algebra}
Given the definition of a relational Haar system, we now define a generalization of the convolution algebra to the case of relational groupoids. To define a set of functions on which the convolution product is well-defined, we will from now on assume that our relational groupoids are equipped with a topology.  A relational Haar systems is assumed to be continuous\footnote{One way to phrase this is that for every continuous function $f$ on $\mathcal{G}^{(3)}$ the function $k \mapsto \int_{(g,h) \in \mathcal{G}^{(2)}}f(g,h,k)d\mu_k$ is continuous} and given by Radon measures with respect to this topology. The following subspace of functions is a convenient one: 
\begin{defn}[Admissible functions]
Let $\mathcal{G}$ be a topological relational groupoid and $\mu= (\mu_g)_{g \in \mathcal{G}}$ a Haar system given by Radon measures. Denote $\underline{\calG}$ the quotient groupoid of $\mathcal{G}$ equipped with the quotient topology. Then a continuous function $f$ is \emph{admissible} if it is bounded and there is a compact set $\underline{\calK} \subset \underline{\calG}$ such that $\mathrm{supp}(f) \cap \mathcal{C} \subset q^{-1}(\underline{\calK})$. The set of admissible functions is denoted by $\mathcal{A}(\mathcal{G})$.
\end{defn}
The point about admissible functions is that they have a well-defined convolution product which is again an admissible function. 
\begin{prop}
Let $f_1,f_2 \in \mathcal{A}(\mathcal{G})$. Then the convolution product defined by \begin{equation}
    (f_1\star f_2)(g) = \int_{\mathcal{G}^{(2)}_g}f_1(h)f_2(k)d\mu_g(h,k)\label{eq:defprod}
    \end{equation}
converges for all $g \in \mathcal{G}$ and $g \mapsto (f_1 \star f_2)(g) \in \mathcal{A})(\calG)$. 
\end{prop}
\begin{proof}
First, notice that $(f_1 \star f_2)(g)=0$ if $g \notin \mathcal{C}$ - this follows from the fact that $\mathcal{G}^{(2)}_g = \emptyset$ in this case (see Proposition \ref{prop:ginC}). Otherwise, we use the axiom that the Haar measure $\mu_g$ disintegrates to write 
\begin{equation}
(f_1 \star f_2)(g) = \int_{(\underline{g}_1,\underline{g}_2) \in \underline{\calG}^{(2)}_{q(g)}}\underbrace{\int_{q^{-1}((\underline{g}_1,\underline{g}_2))}f_1(h)f_2(k)d\mu_{\underline{g}\underline{g}_1\underline{g}_2}(h,k)}_{=:\tilde{f}(\underline{g}_1,\underline{g}_2)}d\nu_{\underline{g}}(\underline{g}_1,\underline{g}_2)\end{equation}
Since $f_1$ and $f_2$ are bounded, integrating $f_1\cdot f_2$ along the fibers results in a bounded continuous function $\tilde{f}$ on the base $\underline{\calG}^{(2)}_{q(g)}$ and $|\tilde{f}|\leq |f_1||f_2|$. Let $\underline{\calK}_i$ denote compact sets containing $q(\mathrm{supp}(f_i)).$ The support of $\tilde{f}$ is contained in $\underline{\calK}_1 \times \underline{\calK}_2$, which is compact. Since $\mu_q(g)$ is also a Radon measure, the integral is finite and in fact $|(f_1 \star f_2)(g)|\leq \nu_{\underline{g}}(\underline{\calK}_1 \times \underline{\calK}_2)|f_1 \times f_2|$. Hence $f_1 \star f_2$ is defined pointwise. By continuity of the Haar system, it follows that $f_1 \star f_2$ is also continuous. It remains to check that $f_1 \star f_2$ is admissible. Let $m$ denote the multiplication in the quotient groupoid, and denote $\underline{\calK}$ the set $m((\underline{\calK}_1 \times \underline{\calK}_2)\cap \underline{\calG}^{(2)})$. Then $\underline{\calK}$ is compact since the multiplication is continuous. The arguments above imply that $\mathrm{supp}(f_1 \star f_2) \subset q^{-1}(\underline{\calK})$, and hence in the preimage of a compact set. To see that $f_1 \star f_2$ is bounded, note that $|(f_1 \star f_2)|\leq |f_1||f_2| \mathrm{sup}_{\underline{g}\in \underline{\calK}}  \nu_{\underline{g}}(\underline{\calK})$. Again by compactness, the supremum is obtained and $f_1 \star f_2$ is bounded, and hence admissible.
\end{proof}
The definition of the convolution algebra is now straightforward:
\begin{defn}
Let $(\mathcal G,L,I)$ be a relational groupoid with a relational Haar system $(\mu_g)_{g\in\mathcal{G}}$. 
Denote by $\mathcal{A}(\mathcal{G})$ the space of admissible functions. The relational convolution algebra is then $(\mathcal{A}(\mathcal{G}),\star)$ with the convolution product defined in Equation \eqref{eq:defprod}.
\end{defn}
%\marginpar{Notation?Assumptions on functions?
%It is not clear that the algebra is closed under convolution this way. }
%    $(\mathcal{A}(\mathcal{G},\mu_g),\star)$

%\begin{equation}
 %   (f_1\star f_2)(z) = \int_{\mathcal{G}^{(2)}_z}f_1(x)f_2(y)d\mu_z(x,y)
  %  \end{equation}
%for $f_1,f_2\in \calA(\calG,\mu_g)$.
%\end{defn}

%\begin{multline}
 %   \int_{\mathcal{G}^{(2)}_z}f_1(x)f_2(y)d\mu_z(x,y) = \int_{(\underline{x},\underline{y}) \in G^{(2)}_{\underline{z}}} \int_{(x,y) \in q^{-1}_g(\underline{x},\underline{y})} f_1(x)f_2(y)d(\mu_z)_{\underline{x}\underline{y}} d\underline{\mu}_{\underline{z}} \\= \int_{(\underline{x},\underline{y}) \in G^{(2)}_{\underline{z}}}\int_{(x,y) \in q^{-1}_g(\underline{x},\underline{y})}d(\mu_z)_{\underline{x}\underline{y}}f_1(\underline{x})f_2(\underline{y})d\underline{\mu}_{\underline{z}} = \underline{f}_1\star\underline{f}_2(\underline{z})
%\end{multline}
\begin{rem}
Another possible domain for the convolution product is given by continuous functions which are just compactly supported. This is a subspace of the space of admissible functions. However, below we want to consider $L_2$-invariant functions, to recover the algebra of functions on the quotient. If the $L_2$-fibers are not compact (for instance, this happens for the relational group associated to $\calG = \mathbb{Z},\calH= k\mathbb{Z}$), then such functions will never be compactly supported. This motivates our choice of space of admissible functions. Of course, for compact relational groupoids the spaces of admissible, compactly supported, and continuous functions all coincide. 
\end{rem}
%\textcolor{blue}{Add the discussion on the admissible functions and the Haar system on relational groups, eith the example of $\mathbb R^2$}
\subsection{Associativity}
A natural question when defining an algebra is whether or not it is associative. This is the question we investigate in this subsection. It turns out that in general, the convolution is associative only when restricted to $L_2$-invariant functions, or when the Haar system satisfies some restrictive condition. 
\begin{prop}\label{prop:L2ass}
The space of $L_2$-invariant functions is an associative subalgebra of $(\mathcal{A}(\mathcal{G}),\star)$.
\end{prop}

Notice that for an $L_2$-invariant function $f$ there is a well-defined continuous function $\tilde{f}$ on the quotient groupoid $\underline{\mathcal{G}}$, defined by $\tilde{f}(\underline{g}) = f(g)$. By Axiom (\ref{item:haarsys2}) of Definition \ref{def:RelHaarSys}, to a Haar system $\mu$ on $\mathcal{G}$ there is an associated Haar system $\nu$ on $\underline{\mathcal{G}}$. Let us denote its convolution product on compactly supported functions by $\star_{\underline{\mathcal{G}}}$. The proof of Proposition \ref{prop:L2ass} is a direct corollary of the following Lemma: 
\begin{lem}\label{lem:L2conv}
For $L_2$-invariant functions $f_1,f_2 \in \mathcal{A}(\mathcal{G})$, the convolution is given by
\begin{equation}
\label{eq:L2conv}
f_1 \star f_2= \begin{cases} q^*(\tilde{f}_1 \star_{\underline{\mathcal{G}}} \tilde{f}_2)&\text{on}\,\,\, \mathcal{C}\\ 0&\text{on}\,\,\, \mathcal{G}\setminus\mathcal{C}\end{cases}
\end{equation}
%\textcolor{purple}{and $f_1 \star f_2 = 0$ on $\mathcal{G} \setminus \mathcal{C}$.}
\end{lem}
\begin{proof}
First, we recall that by definition the support of the convolution of any two admissible functions is contained in $\mathcal{C}$, since $\calG^{(2)}_g = \emptyset$ for $g \in \calG$.
Notice further that if $f_1,f_2 \in \mathcal{A}(\mathcal{G})$ are $L_2$-invariant then by definition their associated functions on $\underline{\mathcal{G}}$ are compactly supported, hence the right hand side of Equation \eqref{eq:L2conv} is defined. Now, we use again the fact that the Haar system disintegrates to write 
\begin{align*}
(f_1 \star f_2)(g) &= \int_{(\underline{g}_1,\underline{g}_2) \in \underline{\mathcal{G}}^{(2)}_{q(g)}}\int_{q^{-1}((\underline{g}_1,\underline{g}_2))}f_1(h)f_2(k)d\mu_{\underline{g}\underline{g}_1\underline{g}_2}(h,k)d\nu_{\underline{g}}(\underline{g}_1,\underline{g}_2)  \\
&= \int_{(\underline{g}_1,\underline{g}_2) \in \underline{\mathcal{G}}^{(2)}_{q(g)}}\tilde{f}_1(\underline{g}_1)\tilde{f}_2(\underline{g}_2) d\nu_{\underline{g}}(\underline{g}_1,\underline{g}_2) \\
&= (\tilde{f}_1 \star_{\underline{\mathcal{G}}} \tilde{f}_2) (\underline{g}).
\end{align*}
Here, the second equation follows from normalization of the disintegrating measures and the fact that $L_2$-invariance is equivalent to being constant on fibers of $q$. Integrating a constant function against a probability measure, we simply obtain that constant. 
\end{proof}
By associativity of the convolution in the quotient groupoid, we obtain Proposition \ref{prop:L2ass}. \\
Next, we give a sufficient criterion on the Haar system for the convolution algebra to be associative:
\begin{prop}
Suppose that the relational Haar system $\mu$ is 
%$L_2$-invariant (Definition \ref{def:l2inv}) and split (Definition \ref{def:split}). 
strongly split as in Definition \ref{def:strongsplit}.
Then the convolution algebra $\mathcal{A}(\mathcal{G})$  is associative. 
\end{prop}
\begin{proof}
Another way to write formula \eqref{eq:L2conv} is 
\begin{equation}
f_1 \star f_2 = q^*(q_*f_1 \star_{\underline{\mathcal{G}}} q_*f_2)\label{eq:pushforwardconv}
\end{equation}
where we denoted $\tilde{f} =: q_*f$ the pushforward of $L_2$-invariant functions. Associativity of the convolution restricted to $L_2$-invariant functions follows from the property \begin{equation} q_* \circ q^* = \mathrm{id}\label{eq:prop}
\end{equation}
i.e. $q_*$ is defined as the left inverse of the injective map $q^*\colon C_c(\underline{\mathcal{G}}) \to \mathcal{A}(\mathcal{G})$ on its image. Thus, the convolution is associative on all function is we can find an extension of the map $q_*$ to all of $\mathcal{A}(\mathcal{G})$ in such a way that property \eqref{eq:pushforwardconv} is still true.  In this case, there is a family of probability measures $\tau_{\underline{g}}$ on $\underline{\mathcal{G}}$ and we can extend $q_*$ by setting 
$$q_*f (\underline{g}) = \int_{q^{-1}(\underline{g})}f(g)d\tau_{\underline{g}}(g)$$
and again, we have properties \eqref{eq:pushforwardconv} and \eqref{eq:prop}. Thus associativity again follows from associativity of the convolution algebra of the quotient.
\end{proof}

\begin{rem}
In the two examples above, we were able to infer associativity from the fact that the convolution algebra on the quotient is associative. However, both of these examples are ``well behaved'' with respect to the quotient, since they satisfy the strongly split condition. Already on the simplest relational groupoid with nontrivial $L_2$ relation, one can find examples of (split, invariant) relational Haar systems such that the convolution algebra is not associative, see Example \ref{ex:non-ass} below. However, since the $L_3$ relation in the relational groupoid is strictly associative, and $L_2$-invariant, one might ask whether there is a remnant of this fact that is visible at the level of the convolution algebra, i.e. whether the algebra is ``associative up to homotopy'' in general. This is an interesting question that the authors plan to address in a future paper.
\end{rem}

The relational convolution algebra associated to a split $L_2$-invariant relational Haar system which is not strongly split is not necessarily associative, as follows from the counterexample below.
\begin{ex}\label{ex:non-ass}
We consider again the relational group given by $\calG = \mathbb{Z}_4, \calH= \mathbb{Z}_2$. Denote $\delta_g$ the function with $\delta_g(g) = 1$ and $\delta_g(h) = 0$ for $h \neq g$. Unraveling the definitions, we can compute 
\begin{align*}
    \delta_0 \star \delta_0 &= \mu_0(0,0)\delta_0 + \mu_2(0,0)\delta_2 \\
    \delta_0 \star \delta_1 &= \mu_1(0,1) \delta_1 + \mu_3(0,1)\delta_3 \\ 
    \delta_2 \star \delta_1 &= \mu_1(2,1) \delta_1 + \mu_3(2,1)\delta_3 \\ 
    \delta_0 \star \delta_3 &= \mu_1(0,3) \delta_1 + \mu_3(0,3) \delta_3 
\end{align*}
\begin{multline*}
    \delta_0 \star (\delta_0 \star \delta_1) = (\underbrace{\mu_1(0,1)^2 + \mu_3(0,1)\mu_1(0,3)}_A)\cdot\delta_1 \\+ (\underbrace{\mu_1(0,1)\mu_3(0,1) + \mu_3(0,1)\mu_3(0,3)}_B)\cdot \delta_3 
\end{multline*}
\begin{multline*}
    (\delta_0 \star \delta_0) \star \delta_1 = (\underbrace{\mu_0(0,0)\mu_1(0,1) + \mu_2(0,0)\mu_1(2,1)}_C)\cdot\delta_1 \\+ (\underbrace{\mu_0(0,0)\mu_3(0,1)  + \mu_2(0,0)\mu_3(2,1)}_D)\cdot \delta_3
\end{multline*} 
The condition for $\mu$ being a relational Haar system implies that $$\mu_g(g_1,g_2) = \nu_{\ul{g}}(\mu_g)_{\ul{g}_1\ul{g}_2}(g_1,g_2) = \frac12(\mu_g)_{\ul{g}_1\ul{g}_2}(g_1,g_2)$$
where $(\mu_g)_{\ul{g}_1\ul{g}_2}(g_1,g_2)$ is some probability measure on $q^{-1}(\ul{g}_1,\ul{g}_2)$. It is clear that in general $A,B,C,D$ are pairwise different. If we suppose the measure is $L_2$-invariant, we obtain 
\begin{align*}
    A &= B = \mu_1(0,1)^2 + \mu_1(0,1)\mu_1(0,3), \\
    C &= D = \mu_0(0,0)(\mu_1(0,1)+ \mu_1(2,1).
\end{align*}
Furthermore, if we suppose that $\mu$ is split, we obtain $\mu_0(0,0) =\frac14 \tau^0_{\ul{0}}(0)^2$. Letting $\tau^0_{\ul{0}}$ be the measure on $\{0,2\}$ supported at 2, we see that $C=D=0$ (in fact in this case $\delta_0\star\delta_0 = 0$). On the other hand, $A=B$ depends only on $\tau^1_{\ul{0}}$ and $\tau^1_{\ul{1}}$. Thus, this is an example of a split $L_2$-invariant Haar system with non-associative convolution algebra.
\end{ex}

\subsection{Reduction of the algebra of admissible functions}
In this subsection we show that the reduction of the relational convolution algebra is isomorphic to the convolution algebra of the groupoid which is the reduction of the relational groupoid.
We will obtain this result via a two-step reduction: one reduction with respect to the constraint set, followed by a reduction with respect to the $L_2$ relation.
Let $(\mathcal G,L,I)$ be a relational groupoid, and let $\mathcal C$ be its constraint set.
We denote by $I_{\mathcal C}$ the subset of functions in  $C_c(\mathcal{G})$ that vanish on $\mathcal C$. The set $I_\calC$ is usually called the vanishing ideal, since it is an ideal inside $C_c(\calG)$ with the standard (commutative) product. Furthermore we obtain the following isomorphism of convolution algebras.
\begin{prop}
$I_{\mathcal C}$ is also an ideal in $\mathcal A (\mathcal G)$ and 
\[\mathcal A(\mathcal{G})/I_{\mathcal C} \cong \mathcal A(\mathcal C).\]
\end{prop}
\begin{proof}
The fact that $I_{\mathcal C}$ is an ideal of $\mathcal A (\mathcal G)$ follows from the fact that if $f_1\in I_{\mathcal C}$, then $f_1$ vanishes on $\mathcal C$, and hence $f_1\star f_2\vert_{\mathcal C}=0$, for all $f_2\in \mathcal A(\mathcal G)$.
\end{proof}

\begin{defn}Denote by  $(\mathcal A(\mathcal C))^{L_2}\subset \mathcal A(\mathcal C)$ the subspace of admissible functions on $\mathcal C$ that are constant along $L_2$, i.e. the subspace of functions $f \in \mathcal A(\mathcal C)$ satisfying $(h,k) \in L_2 %\cap \mathcal C \times \mathcal C 
\Rightarrow f(h) = f(k)$.
\end{defn}
\begin{prop}\label{L2conv}
$(\mathcal A(\mathcal C)^{L_2},\star)$ is a subalgebra of $(\mathcal A (\mathcal C),\star)$.
\end{prop}
\begin{proof}
It follows directly from Lemma \ref{lem:L2conv}.
\end{proof}
\begin{defn}[Reduced convolution algebra]
Let $\mathcal{G}$ be a relational groupoid. Its \emph{reduced convolution algebra} (with respect to $L_2$) is 
\begin{equation}
    \underline{\mathcal A}(\mathcal{G}) =\left(\mathcal A(\mathcal{G})/I_{\mathcal C}\right)^{L_2}\cong \mathcal A(\mathcal C)^{L_2}.
\end{equation}
\end{defn}
We are now ready to state our main result. 
%(wait for it...) 
\begin{thm}
The reduced convolution algebra of a relational groupoid is isomorphic to the (groupoid) convolution algebra of its reduction: 
\begin{equation}
    \underline{\mathcal A}(\mathcal{G}) \cong \mathcal A(\underline {\mathcal G}).
\end{equation}
\end{thm}
%TATATATAAA! 
\begin{proof}
%By pure awesomeness! \textcolor{magenta}{I like this proof :-), we should leave this.}
Since $\underline {\mathcal G}= \mathcal C / L_2$, it follows that the sets $\underline{\mathcal A}(\mathcal{G})$ and $\mathcal A(\underline {\mathcal G})$ are in bijection. The fact that the map
\begin{eqnarray*}
\Phi: \underline{\mathcal A}(\mathcal{G}) &\to& \mathcal A(\underline {\mathcal G})\\
f&\mapsto& \underline f
\end{eqnarray*}
is an isomorphism of convolution algebras follows from $L_2$-invariance and Proposition \ref{L2conv}.
\end{proof}

\section{Future Directions}
\label{sec:Future_Directions}
\subsection{Relational $C^*$-algebras and higher structures.} A natural next step is to introduce the $C^*$-completion of the relational convolution algebra. In particular, we will study the relational analogue of the algebra $M(n,\mathbb C)$ of complex valued $n\times n$-matrices and the algebra $B(H)$ of bounded linear operators on a complex Hilbert space $H$. We will also study the representations of relational convolution algebras and how this construction relates to $C^*$-algebras for 2-groups and higher-categorical versions of relational groupoids.

\subsection{Field theory}\label{Field Theory} In a follow-up work, we intend to describe the groupoid convolution algebra for the relational symplectic groupoid obtained via the Poisson Sigma Model. This would provide a positive answer for the Guillemin--Sternberg conjecture for this particular 2-dimensional TFT. In particular, we will use Hawkins' approach to $C^*$-algebra quantization of symplectic groupoids \cite{Hawkins}, in the context of relational groupoids. 

More precisely, Hawkins developed a $C^*$-algebra quantization procedure that is compatible with the groupoid structure maps, as well as with the multiplicative symplectic structure on the space of morphisms. The construction is based on prescribing the following data (each built based on the previous ones):
\begin{enumerate}
    \item A symplectic groupoid $G\rightrightarrows M$ integrating a Poisson manifold $M$.
    \item A prequantum line bundle $(L,\nabla)$ over $G$.
    \item A symplectic groupoid polarization $\mathcal P$ of $G$.
    \item A ``half-form" bundle $\Omega_{\mathcal P}^{\frac{1}{2}}$.
    \item A twisted, polarized convolution algebra $\mathcal C^*_{\mathcal P}(G, \sigma)$ (where $\sigma$ encodes the twisting data).
\end{enumerate}
Our construction serves as a first step towards a field theoretic interpretation of the geometric quantization procedure introduced by Hawkins. In particular, we conjecture that the final object in step (5) can be obtained via transgression in the AKSZ formulation of the Poisson Sigma Model. A natural candidate for the relational Haar system comes from the path space construction, analogous to the Wiener measure on loop groups.
This interpretation will help connect the perturbative quantization of the Poisson Sigma Model via relational symplectic groupoids \cite{RSG_Quantization} and non perturbative approaches, i.e. geometric quantization 
\cite{PSM_GQ}. Such interpretation will be independent of the integrability of the Poisson manifold.

\subsection{Relational convolution and split relations} In \cite{Split}, split relations were studied in the context of relational symplectic groupoids. A relation is called \emph{split} if it is isotropic and it has a closed isotropic complement. It turns out that the structure relations of every finite-dimensional symplectic groupoid is split, as well as the infinite-dimensional relational symplectic groupoid obtained via the PSM. We will study the relationship between split relations and the convolution algebra for split relational groupoids.

\subsection{Relational convolution and Frobenius algebras} 
In \cite{Frobenius1}, it has been proven that special dagger Frobenius objects in \textbf{Rel} (the category of sets and relations) are in one-to-one
correspondence with groupoids. Recently \cite{Frobenius2}, this result has been generalized using a characterization of Frobenius objects in \textbf{Rel} using simplicial sets. 
Relational convolution algebras are natural examples of Frobenius objects in \textbf{Rel}, in the same way that group algebras are a special class of Frobenius algebras. We expect to have a simplicial interpretation of these extension in the future \cite{MehtaKeller}.

\appendix
%\textcolor{magenta}{We should fix the notation in the whole appendix.}

\section{Some results on relational groupoids}

In this appendix we prove some results about the structure of the sets 
\[\calG^{(2)}_g =\{(h,k)\in \calG^{(2)}\mid (h,k,g) \in L_3\}.\]
Moreover, let $\calC$ be defined as in Proposition \ref{prop:constraint}. Then we have the following proposition.
\begin{prop}\label{prop:ginC}
If $\calG^{(2)}_g$ is nonempty then $g \in \calC$.
\end{prop}
\begin{proof}
Let $(h,k) \in \calG^{(2)}_g$. Then $(h,k,g) \in L_3 = L_2 \circ L_3$ by $L_2$-invariance \eqref{invariance3}. It follows that there is $g'$ such that $(g',g) \in L_2$. In particular, $g \in L_2 \circ \mathcal{G} = \calC$. 
\end{proof}
For a usual groupoid $G\rightrightarrows M$, these sets are the fibers of the fibration $G^{(2)} \to G$. In a relational groupoid, this is no longer true. Instead we have the following statement. 
\begin{prop}
For two distinct elements $g,g'\in \mathcal G$, the sets $\mathcal G^{(2)}_g$ and $\mathcal G^{(2)}_{g'}$ are equal if and only if $(g,g') \in L_2$ and disjoint otherwise. 
\end{prop}
\begin{proof}
First, suppose that $(g,g') \in L_2$ and let $(h,k,g) \in L_3$. Then $(h,k,g') \in L_2 \circ L_3 = L_3.$  Hence $\mathcal G^{(2)}_g \subseteq \mathcal G^{(2)}_{g'}$, and symmetry of $L_2$ implies the other direction. 
Now, assume that $(h,k) \in \mathcal G^{(2)}_g \cap \mathcal G^{(2)}_{g'}$. It suffices to show that $(g,g') \in L_2$. Axioms \underline{\textbf{A.1}} and \underline{\textbf{A.3}} as in Definition \ref{defn:relational_groupoid} imply that $(h,k,g) \in L_3$ if and only if $(g,I(k),h) \in L_3$. Hence, 
$(g,I(k),k,g') \in L_3 \circ (L_3 \times \id) = L_3 \circ (\id \times L_3)$. Hence, there is $e \in \mathcal G$ such that $(I(k),k,e) \in L_3$ and $(g,e,g') \in L_3$, by definition of $L_1$ and $L_2$, we conclude that $e \in L_1$ and hence $(g,g') \in L_2$.
\end{proof}
An immediate corollary is the following.
\begin{cor} Let $(h,k) \in \mathcal G^{(2)}_g$, then the set $hk := \{g' \in \mathcal G\mid (h,k,g') \in L_3\}$ satisfies $hk = L_2(g)$. In particular, we have 
\begin{equation}
\mathcal G^{(2)}_{hk} =  \bigcup_{g' \in hk}\mathcal G^{(2)}_{g'} = \mathcal G^{(2)}_g.
\end{equation} 
\end{cor}

Analogous to the right relational action, we can define left relational actions. In particular, a relational groupoid acts on itself by left and right multiplication and the sets $\mathcal G^{(2)}_g$ behave nicely under this action. 
\begin{prop}
\label{prop:Left_right_relational_action}
Let $(\mathcal G,L,I)$ be a relational groupoid and suppose that $(g,h) \in \mathcal G^{(2)}$. Then, we have 
\begin{equation}
    (\id \times R_h) \circ \mathcal G^{(2)}_g = \mathcal G^{(2)}_{gh}, 
\end{equation}
and similarly 
\begin{equation}
    (L_g \times \id) \circ \mathcal G^{(2)}_h = \mathcal G^{(2)}_{gh}, 
\end{equation}
\end{prop}
\begin{proof}
For the first equation, we have 
\begin{align*} 
(\id \times R_h)  \circ \mathcal G^{(2)}_g &=  \{(g_1,g_2)\mid \exists k, (g_1,k) \in G^{(2)}_g, (k,g_2) \in R_h,\} \\
&=  \{(g_1,g_2)\mid \exists k, (g_1,k) \in \mathcal G^{(2)}_g, (g_2,k) \in (R_h)^T = R_{I(h)},\} \\
&=  \{(g_1,g_2)\mid \exists k, (g_1,k,g) \in L_3, (g_2,I(h),k) \in L_3\} \\
&=  \{(g_1,g_2)\mid (g_1,g_2,I(h),g) \in L_3 \circ (\id \times L_3)\} \\
&=  \{(g_1,g_2)\mid (g_1,g_2,I(h),g) \in L_3 \circ (L_3 \times id)\} \\
&=  \{(g_1,g_2)\mid \exists k \in \mathcal G\colon(g_1,g_2,k) \in L_3, (k,I(h),g) \in L_3\} \\
&=  \{(g_1,g_2)\mid \exists k \in \mathcal G\colon(g_1,g_2,k) \in L_3, (k,I(h),I(g)) \in L \} \\
&=  \{(g_1,g_2)\mid \exists k \in \mathcal G\colon(g_1,g_2,k) \in L_3, (I(h),I(g),k) \in L \} \\
&=  \{(g_1,g_2)\mid \exists k \in \mathcal G\colon(g_1,g_2,k) \in L_3, (I(h),I(g),I(k)) \in L_3 \} \\
&=  \{(g_1,g_2)\mid \exists k \in \mathcal G\colon(g_1,g_2,k) \in L_3, (g,h,k) \in L_3 \} \\
&=  \mathcal G^{(2)}_{gh}
\end{align*}
The other equation is proven similarly.
\end{proof}
In particular, we have the following. 
\begin{cor}
Let $(\mathcal G,L,I)$ be a relational groupoid and suppose that $(g,h) \in \mathcal G^{(2)}$. Then 
\begin{equation}
    (L_{I(g)} \times R_h) \circ \mathcal G^{(2)}_g = \mathcal G^{(2)}_h
\end{equation}
\end{cor}
\begin{proof}
By Proposition \ref{prop:Left_right_relational_action}, we have $(L_{I(g)} \times R_h) \circ \mathcal G^{(2)}_g = \mathcal G^{(2)}_{I(g)gh}.$ By definition of $L_2$ this is can be rewritten as $\mathcal G^{(2)}_{I(g)gh}= \mathcal G^{(2)}_{L_2(h)}$. We conclude the proof using Proposition \ref{prop:Left_right_relational_action}. 
\end{proof}

\end{document}